\documentclass[lettersize,journal]{IEEEtran}
\IEEEoverridecommandlockouts
\ifCLASSINFOpdf
\else
\fi

\usepackage{amsthm}  

\usepackage{tabularx}
\usepackage{threeparttable}
\usepackage{mathrsfs}
\usepackage{color}
\usepackage{amsfonts,amssymb}
\usepackage{amsfonts}
\usepackage{subfigure}
\usepackage{booktabs}

\usepackage{amsmath}
\usepackage{enumerate}
\usepackage{algorithm}
\usepackage{algorithmic}
\usepackage{bm}
\usepackage{makecell}
\usepackage{multirow}
\usepackage{mathtools}
\usepackage{bbm}
\usepackage{dsfont}
\usepackage{pifont}

\usepackage{cite}
\usepackage{balance}
\usepackage{enumerate}

\def\BibTeX{{\rm B\kern-.05em{\sc i\kern-.025em b}\kern-.08em
    T\kern-.1667em\lower.7ex\hbox{E}\kern-.125emX}}

\newtheorem{remark}{Remark}
\newtheorem{theorem}{Theorem}
\newtheorem{lemma}{Lemma}
\newtheorem{corollary}{Corollary}

\newtheorem{definition}{Definition}

\begin{document}

\title{Inverse Inference on Cooperative Control of Networked Dynamical Systems}
\author{Yushan Li$^{\dag}$, Jianping He$^{\ddag}$,  and Dimos V. Dimarogonas$^{\dag}$
\thanks{
    ${\dag}$: Yushan Li and Dimos V. Dimarogonas are with the Division of Decision and Control Systems, KTH Royal Institute of Technology, Stockholm, Sweden. E-mail: \{yushanl,dimos\}@kth.se.}
\thanks{ 
    ${\ddag}$: Jianping He is with the Department of Automation, Shanghai Jiao Tong University, Shanghai, China. E-mail: jphe@sjtu.edu.cn. }%
}

\maketitle

\begin{abstract}
Recent years have witnessed the rapid advancement of understanding the control mechanism of networked dynamical systems (NDSs), which are governed by components such as nodal dynamics and topology. 
This paper reveals that the critical components in continuous-time state feedback cooperative control of NDSs can be inferred merely from discrete observations. 
In particular, we advocate a bi-level inference framework to estimate the global closed-loop system and extract the components, respectively. 
The novelty lies in bridging the gap from discrete observations to the continuous-time model and effectively decoupling the concerned components. 
Specifically, in the first level, we design a causality-based estimator for the discrete-time closed-loop system matrix, which can achieve asymptotically unbiased performance when the NDS is stable. 
In the second level, we introduce a matrix logarithm based method to recover the continuous-time counterpart matrix, providing new sampling period guarantees and establishing the recovery error bound. 
By utilizing graph properties of the NDS, we develop least square based procedures to decouple the concerned components with up to a scalar ambiguity. 
Furthermore, we employ inverse optimal control techniques to reconstruct the objective function driving the control process, deriving necessary conditions for the solutions. 
Numerical simulations demonstrate the effectiveness of the proposed method.   
\end{abstract}


\section{Introduction}\label{sec:introduction}
In the last decades, networked dynamical systems (NDSs) have played a crucial role in many engineering and biological fields, e.g., multi-robot formation \cite{oh2015survey}, power grids \cite{10113752}, human brain \cite{he2022uncovering}, and immune cell network \cite{tajvar2023modelling}. 
An NDS, comprising multiple interconnected nodes, is characterized by not only the self-dynamics of a single node (nodal dynamics) but also the interaction structure (topology) between nodes, 
and can achieve various cooperative behaviors such as synchronization. 
However, the prior information about the nodal dynamics and topology is not always accessible in practice, and needs to be inferred from observations. 
This inference enhances our ability to understand, predict, and intervene with the NDS \cite{gao2023data}.

\subsection{Motivations}

This paper focuses on the continuous-time linear state-feedback cooperative control of NDSs, where only discrete and noisy observations on a single round of the system's trajectory are available. 
In particular, we aim to provide a systematic approach for inferring both the nodal dynamics and the topology of the NDS, and subsequently reconstructing the control objective function that governs the process. 

The motivation for addressing this problem stems from two main aspects. 
First, most existing studies that concentrate on network inference or identification primarily deal with the topology structure, while assuming that the nodal dynamics are known a priori and that the nodal state is scalar \cite{vanwaarde2021topology,LI20238381}. 
However, finding solutions for cases involving high-dimensional and unknown nodal dynamics remains an open issue. 
Second, in contrast to traditional system identification methods, which typically operate under open-loop testing on a single system, the state feedback control of NDSs is an interconnected and closed-loop process without any external excitation. 
This characteristic implies that the nodal dynamics and topology are inherently intertwined with the underlying feedback gain, creating an additional challenge in decoupling these unknown factors. 
Therefore, it is essential to develop new methodologies to simultaneously tackle such issues and provide theoretical guarantees.


\subsection{Related Works}

The study of inferring the cooperative control of NDSs has a close relationship with multiple fields, including system identification, inverse optimal control, and topology inference. 
We briefly review the literature most related to this work as follows.

Inferring or identifying a system model from input/output data has a long history, especially for linear time-invariant systems. 
The basic idea is to let the systems run open-loop or excite them with inputs and use the collected to construct estimators. 
Numerous classic works have been devoted to the asymptotic consistency of the proposed estimators \cite{ljung1998system}. 
Recently, the finite time identification performance on the system model such as the Markov matrix has received considerable attention \cite{simchowitz2018learning,sarkar2019near,zheng2021nonasymptotic,9946382}, which mainly depends on the tools of stochastic optimization and concentration measure (we refer to \cite{10383489} for a detailed review). 
Note that in these works, the system implementation for identification is determined by the user to collect sufficient informative data (e.g., injecting random inputs into the system), and does not involve closed-loop control with specific objectives. 
In addition, there is a growing number of studies focusing on ``learning'' a controller from the collected data without the need for system identification \cite{fattahi2020efficient,depersis2021lowcomplexity,hu2023theoreticala}. 
For example, the work \cite{9903320} develops an effective data-driven resilient control scheme with the system stability guaranteed. 
The authors in \cite{yu2023onlineb} propose an online algorithm to stabilize an unknown networked linear system under communication constraints and disturbance. 
Concerning this direction, we observe that the focus of the aforementioned works is to construct appropriate stabilizing controllers for the unknown system, not to infer the underlying mechanism of well-designed controllers.

Given the observations of the system trajectory, the inverse optimal control provides an effective paradigm to infer the objective functions that determine the control task \cite{abazar2020inverse}. 
In this direction, the seminal work \cite{Kalman1964optimal} proposed the inverse optimality problem for linear systems. 
Then, numerous works have been developed based on different system models and prior assumptions, e.g., linear \cite{priess2015solutions} or nonlinear systems \cite{menner2020maximum,jin2021inverse}. 
Specifically, the works \cite{zhang2019inverseb,li2020continuoustime,yu2021system} investigate the popular linear quadratic (LQ) regulator cases, where the feedback gain matrix or the optimal inputs are required to be given.  
The recent works \cite{zhang2024statistically,10697272} have further shown that the objective function can be recovered by merely using multiple observation trajectories. 
Despite the fruitful results, we observe that few works have considered the inverse optimal control of networked dynamical systems, where the topology structure will pose new challenges to the algorithm design.


Note that the network topology and node dynamics intertwine in an NDS. 
Towards the former aspect, numerous methods have been developed on topology identification in the literature, e.g., by graph signal processing \cite{mateos2019connecting}, vector autoregression \cite{zaman2021online}, compressed sensing \cite{hayden2016sparse}, excitation design \cite{bombois2023informativity}, etc. 
Note that in most topology inference literature, the nodal dynamics are assumed as prior knowledge, e.g., see \cite{vanwaarde2021topology,prasse2022predicting,10337619}. 
A few works have also investigated the algorithm design for joint inference situations. 
For example, 
\cite{wai2019joint} investigates how to jointly estimate the network topology and the parameters of nodal dynamics from perturbed stationary points, 
and \cite{10210488} utilize the backpropagation through time algorithm to construct a neural network, which is further used to infer the topology structure and nodal dynamics. 
The nodal dynamics functions in these works can be nonlinear, but the nodal states are generally one-dimensional. 
In contrast, this work focuses on the linear high-dimensional cases for nodal dynamics, 
which will introduce another unknown feedback gain matrix in the inference process. 
How to overcome the influence of this component needs to be further addressed.

\subsection{Contributions}

To address the aforementioned issues, we propose a bi-level inference method to reveal the components involved in state feedback cooperative control, including the nodal dynamics, the topology, and the feedback control gain. 
The primary challenges lie in i) establishing the explicit relationship between discrete observations and continuous-time model parameters, 
and ii) overcoming the inherent couplings concerned components. 
Our key idea is to first construct the closed-loop system matrix of the NDS from noisy observations, and then extract the relevant components by exploiting the cooperative nature. 
The contributions are summarized as follows.

\begin{itemize}
\item We derive conditions of inferring the control mechanism of NDSs through observation driven methods. 
Precisely, we propose a bi-level inference framework to infer the nodal dynamics and the topology in continuous-time state-feedback cooperative control, along with theoretical guarantees. 

\item  At the first level, we design a causality based estimator for the discretized closed-loop system matrix. 
This estimator effectively bridges the gap between discrete observations and the continuous-time model, achieving unbiased accuracy in the asymptotic sense when the NDS is stable. 

\item We develop matrix logarithm and least squares based procedures to decouple the concerned components in the second level. 
Sufficient conditions regarding the sampling period for continuous-time model estimation are derived, along with the corresponding error bound. 
Notably, the topology and feedback gain can be accurately inferred up to a scalar ambiguity. 

\item Furthermore, we reconstruct the LQ objective function underpinning the state feedback control process. 
By leveraging inverse optimal control techniques, we explicitly characterize the solution space for the parameters. 
Finally, a constrained quadratic optimization problem is formulated to infer the parameters.  
\end{itemize}





The remainder of this paper is organized as follows.
Section \ref{sec:preliminaries} presents some preliminaries and the inference problem formulation. 
The inference design on the nodal dynamics and the topology are provided in Section \ref{sec:infer_control} along with performance analysis. 
The reconstruction method on the objective function is further in Section \ref{sec:infer_objective}. 
Simulation results are given in Section \ref{sec:simulation}. 
Finally, Section \ref{sec:conclusions} concludes the paper.

\section{Preliminaries}\label{sec:preliminaries}

\subsection{Graph Basics and Notations}
A NDS is represented by a directed graph $\mathcal{G}=(\mathcal{V},\mathcal{E})$, where $\mathcal{V}=\{1, \cdots, N\}$ is the node set and $\mathcal{E}\subseteq \mathcal{V}\times \mathcal{V}$ is the set of interaction edges. 
An edge $(i,j)\in \mathcal{E}$ indicates that node $j$ will send information to node $i$.
The adjacency matrix $\textbf{A}_0=[a_{ij}]_{N \times N}$ of $\mathcal{G}$ is defined such that ${a}_{ij}\!>\!0$ if $(i,j)$ exists or ${a}_{ij}\!=\!0$ otherwise. 
Denote by ${\mathcal{N}_i}=\{j\in \mathcal{V}:a_{ij}>0\}$ the in-neighbor set of $i$. 
Let $L=\operatorname{diag}(\textbf{A}_0\bm{1})-\textbf{A}_0$ be the Laplacian matrix of $\mathcal{G}$, and $L\bm{1}=0$ always holds. 
Specially, we consider that $L$ has a diagonal Jordan form, which is commonly used in cooperative optimal control approaches \cite{movric2014cooperative,zhang2015distributed}.

In this paper, we use $\otimes$ to represent the Kronecker product operator. 
For a matrix $M\in \mathbb{R}^{p\times q}$,  we use $\|\cdot\|$, $\|\cdot \|_{F}$, $\sigma_{\max}$ and $\sigma_{\min}$ to represent its spectral norm, Frobenius norm, maximum and minimum singular values, respectively. 
The notation $\operatorname{vec}(M)\in\mathbb{R}^{pq \times 1}$ means the column vector by sequentially stacking the columns of $M$. 
For the matrix equation $M_1 X M_2=M_3$ where the matrices are with compatible dimensions, the vectorization of the equation satisfies $(M_2^\intercal \otimes M_1)\operatorname{vec}(X)=\operatorname{vec}(M_1 X M_2)=\operatorname{vec}(M_3)$. 
A square matrix $M$ is called simple if it has a diagonal Jordan form. 
For two real-valued functions $f_1$ and $f_2$, $f_1(x)=\bm{O}(f_2(x))$ as $x\to x_0$ means $\mathop {\lim }\nolimits_{x \to x_0 } |f_1(x)/f_2(x)|<\infty$. 

\subsection{System Modeling}

Consider $N$ nodes in a NDS $\mathcal{G}=(\mathcal{V},\mathcal{E})$ share a common state space $\mathbb{R}^{n}$ and aim to cooperatively accomplish a task. 
Let $x_i\in\mathbb{R}^{n}$ be the state of node $i$, subject to the following identical linear time-invariant (LTI) dynamics
\begin{equation}
\left\{
\begin{aligned}\label{eq:node_dynamics}
\dot{x}_{i}&=A {x_i} + B {u_i} \\
{u_i}&=K \sum\nolimits_{j \in \mathcal{N}_i } a_{i j} \left(x_{j}-x_{i}\right) 
\end{aligned}
\right.,
\end{equation}
where $A\in\mathbb{R}^{n\times n}$ and $B\in\mathbb{R}^{n\times m}$ are the nodal system matrices, $u_i \in\mathbb{R}^{m}$ is the control input of node $i$, $a_{ij}$ is the weight of the directed edge from node $j$ to $i$, and $K\in\mathbb{R}^{m\times n}$ is the feedback gain matrix. 
Specifically, $(A, B)$ is controllable, and $B$ has full rank. 
The system model \eqref{eq:node_dynamics} is popular in many consensus-based cooperation tasks, e.g., formation forming and tracking control of mobile robots, frequency regulation in smart grids, and oscillatory synchronization in brain neural networks. 
Correspondingly, the global closed-loop dynamics of the NDS is written as 
\begin{equation}\label{eq:global_a}
\dot{x} = (I_{N} \otimes A - L\otimes {B} K)x \triangleq A_c x,
\end{equation}
where $x\!=\![x_1^\intercal,\cdots,x_N^\intercal]^\intercal \!\in\! \mathbb{R}^{Nn}$ is the global state, 
and $A_c$ is called the global continuous-time closed-loop matrix. 
The stability of \eqref{eq:global_a} is guaranteed by the following lemma. 
\begin{lemma}[see \cite{fax2004information}]\label{lemma:consensus}
Let $\lambda_i,i=1,\cdots,N$, be the eigenvalues of $L$. 
The global system \eqref{eq:global_a} achieves consensus if and only if all the matrices
\begin{equation}
A-\lambda_i{B}K,~i=1,\cdots,N,
\end{equation}
are asymptotically stable. 
\end{lemma}

\subsection{Observation Modeling and Problem of Interest}


Consider an external observer can observe the state trajectory of the NDS with a sampling period $\tau>0$ and zero-order holder. 
Denote $x_i(k)$ as node $i$'s state at instant $k \tau$, and the corresponding observation is given by 
\begin{equation}
{y_i}(k)= {x_i}(k)+{\upsilon_i}(k),
\end{equation}
where ${\upsilon_i}(k)\sim \textbf{N}(0, \Gamma)$ is an independent Gaussian noise vector, with $\Gamma=\operatorname{diag}\{ \sigma_{\upsilon,1}^2, \cdots, \sigma_{\upsilon,n}^2\}$ being the covariance matrix for $n$ different state dimensions. 
Then, by integrating \eqref{eq:global_a} and conducting discretization, 
the corresponding discrete-time model that the observer relies on is given by 
\begin{equation}\label{eq:global_discrete}
\left\{
\begin{aligned}
x(k+1)&= e^{A_c \tau} x(k) \triangleq A_d x(k) \\
y(k)&=x(k)+\upsilon(k) 
\end{aligned}
\right.,
\end{equation}
where ${y}\!=\![y_1^\intercal,y_2^\intercal,\cdots\!,y_N^\intercal]^\intercal \!\in\! \mathbb{R}^{Nn}$, $\upsilon\!=\![\upsilon_1^\intercal,\upsilon_2^\intercal,\!\cdots\!,\upsilon_N^\intercal]^\intercal \!\in\! \mathbb{R}^{Nn}$, 
and the matrix exponential $A_d$ represents the discretized closed-loop system matrix. 
Note that $A_d$ is called stable (including marginally and asymptotically stable) if it belongs to the following set 
\begin{align}
&\mathcal{S}_d\!=\!\{ M: \rho(M)\!<\!1,~\text{or 1 is a semisimple eigenvalue}\nonumber \\
&\text{ while other eigenvalue modulus are less than 1}\}.
\end{align}

Based on the above formulation, the goal of the observer is to infer the cooperative control of the NDS from the noisy observations $\{{y}(k)\}_{k=1}^T$. 
Specifically, the components that needs to be inferred consists of three aspects: i) the nodal dynamics matrices $A$ and $B$, ii) the topology matrix $L$ among different nodes, and iii) the local control feedback gain $K$. 
Mathematically, this goal can be formulated as finding an algorithm mechanism $\mathcal{M}$ 
\begin{equation}\label{problem:noise_observation}
\begin{aligned}
 \mathcal{M}(y(k),~k=1,\cdots, T) \mapsto \{A,B,K,L\}. 
\end{aligned}
\end{equation}
Directly solving the above problem is extremely challenging. 
On the one hand, the four matrices are encoded in the continuous-time closed-loop matrix $A_c$, which requires appropriate sampling period setting and algorithm design to reconstruct $A_c$ from the discrete observations. 
On the other hand, the four unknown matrices are coupled with each other in the matrix $A_c$, indicating further decoupling techniques need to be developed. 

With $\mathcal{M}$ obtained, our secondary goal of interest is to infer the objective function parameters behind the cooperative control law \eqref{eq:global_a}, considering the global state-feedback control is designed by optimal LQ control methods. 
Solving this problem will provide another perspective to interpret the cooperation pattern of the NDS from its associated cost functions, e.g., it can reflect the importance of different nodal states and inputs in the control process.


\section{Infer the Cooperative Control Parameters}\label{sec:infer_control}

In this section, we first demonstrate how to design the estimator for the discretized system matrix along with performance analysis. 
Then, we give a sufficient condition for reconstructing the continuous system matrix from the discretized one, and develop a relevant quadratic programming based algorithm. 
Finally, we provide least square based decoupling procedures to obtain the cooperative control parameters, respectively. 
To ease notation in the following estimator design and analysis, we organize the state/observation/noise vectors from time $0$ to $T$ into corresponding matrix forms, given by 
\begin{equation*}
\begin{aligned}
X_T^- &\!=\![x(0),x(1),\!\cdots\!,x(T-1)] ,~X_T^+ \!=\![x(1),x(2),\!\cdots\!,x(T)] , \\
Y_T^- &\!=\![y(0),y(1),\!\cdots\!,y(T-1)] ,~Y_T^+ \!=\![y(1),y(2),\!\cdots\!,y(T)], \\ 
\Phi_T^- &\!=\![\upsilon(0),\upsilon(1),\!\cdots\!,\upsilon(T-1)], ~ \Phi_T^+ \!=\![\upsilon(1),\!\cdots\!,\upsilon(T)]. 
\end{aligned}
\end{equation*}

\subsection{Discretized Closed-loop System Matrix Estimation}

In this part, the discretized system matrix is developed based on the observations, while taking the state evolution characteristic into account. 

Despite the achieved cooperation of the NDS, the state evolution can be varying and complex due to the high-order nodal dynamics. 
Particularly, if the closed-loop matrix $A_d$ is stable (i.e., $A_d\in \mathcal{S}_d$), it follows from using Jordan decomposition on $(A_d)^k$ that the state $x (t)=(A_d)^k x(0)$ will converge to a constant vector exponentially. 
For simplicity, we do not consider very special cases where the configuration of an initial state and an unstable $A_d$ leads to a converging state. 
Then, from the discrete observations, we introduce the following constant cooperation pattern that corresponds to $A_d \in \mathcal{S}_d$. 
\begin{definition}[Constant cooperation pattern]\label{def:constant}
The system \eqref{eq:global_a} is called to achieve a constant cooperation pattern, if given a non-zero bounded initial state, there exists a constant vector $\bm{c}\in\mathbb{R}^{n\times 1}$ such that 
\begin{equation}\label{eq:constant}
\mathcal{C}(x):~\lim_{t \to \infty } \| x_i(t)- \bm{c} \|=0,~\forall i\in\mathcal{V}. 
\end{equation}
\end{definition}

Although in practice whether $A_d\in \mathcal{S}_d$ is not known beforehand, it can be empirically determined from the observations. 
We introduce a time window with length of $L (0<L<T)$ ending at instant $T$, and define the following first-order difference based empirical error
\begin{equation}\label{eq:criteria_error}
\epsilon_1=\sum\limits_{k=T-L}^{T-1} \frac{ \left\| y(k+1)-y(k) \right \| }{L} . 
\end{equation}
Then, the judging indicator for the constant cooperation pattern is given by
\begin{equation}\label{eq:criteria_rho}
\hat{\mathbb{I}}(A_d \in \mathcal{S}_d)=\left\{
\begin{aligned}
&1,&&\text{if}~\epsilon_1 \le \epsilon \\
&0, &&\text{otherwise}
\end{aligned}\right.,
\end{equation}
where $\epsilon>0$ is a preset decision threshold. 
Notice that $\hat{\mathbb{I}}(A_d \in \mathcal{S}_d)$ can effectively determine the stability of $A_d$ with sufficiently long observation horizon $T$. 
As the NDS will converge exponentially to a constant vector, the term $\|y(k+1)-y(k)\|$ can be regarded as the norm of a sum of two independent noises. 
Hence, the value of $\epsilon_1$ can be chosen dependent on the variance of the noise, e.g., setting as $3\sqrt{2}\max\{\sigma_{\upsilon,\ell},\ell=1,\cdots,n\}$ by the famous $3\sigma$-rule \cite{pukelsheim1994three}. 
Inspired by the causality-based topology estimator proposed by \cite{li2023topology}, we define the observation covariance matrix and its one-lag version as
\begin{equation}\label{eq:sample11}
\begin{aligned}
\Sigma_0(T)=\frac{1}{T}(Y_T^-) (Y_T^-)^\intercal,~\Sigma_1(T)=\frac{1}{T}(Y_T^+) (Y_T^-)^\intercal.
\end{aligned}
\end{equation}
Based on the two cases indicated in \eqref{eq:criteria_rho}, we develop the following estimator to infer $A_d$, given by 
\begin{equation}\label{eq:discrete_estimator}
\!\!\widehat{A}_{d}(T) \!=\!\left\{
\begin{aligned}
& \Sigma_1(T) \left( \Sigma_0(T) \!-\!I_N \!\otimes\! \Gamma \right)^{-1},\text{if}~\hat{\mathbb{I}}(A_d \in \mathcal{S}_d)=1\\
& \Sigma_1(T) ( \Sigma_0(T) )^{-1},\text{otherwise}
\end{aligned}\right. .
\end{equation}
The main difference of the two rows in \eqref{eq:discrete_estimator} lies in that if $A_d \in \mathcal{S}_d$ holds, the influence of the observation noises can be effectively eliminated in a de-regularization way (i.e., subtracting $I_N \otimes \Gamma$ from $\Sigma_0(T)$). 
Detailed derivations can be witnessed in the proof of the following Theorem \ref{th:matrix_error}. 

\begin{remark}\label{rema:invertibility}
The key of estimator \eqref{eq:discrete_estimator} lies in the invertibility of $\Sigma_0(T)$. 
Note that the observations $\{ y(k) \}_{k=1}^T$ are corrupted by independent random noises. 
Based on Sard's theorem in measure theory, both $Y_T^-$ and $Y_T^+$ have full rank almost surely. 
Since $I_N \otimes \Gamma$ is a diagonal matrix and we do not consider the case that the state magnitude is in the same scale as the noise magnitude, the invertibility of $(\Sigma_0(T) - I_N \otimes \Gamma)$ is also guaranteed almost surely.  
\end{remark}

With the feasibility of $\widehat{A}_{d}(T)$ guaranteed, we further present its asymptotic performance as the following theorem. 
\begin{theorem}\label{th:matrix_error}
Suppose that the indicator \eqref{eq:criteria_rho} can correctly determine whether $A_d \in \mathcal{S}_d$. 
When the observation horizon $T\to\infty$, the inference error of estimator \eqref{eq:discrete_estimator} satisfies
\begin{equation}
\left\{
\begin{aligned}
&\Pr \left\{\mathop {\lim }\limits_{T \to \infty } \| \widehat{A}_d(T) - {A_d} \|_F^2 \!=\! 0 \right\}\!=\!1,~\text{if}~A_d \in \mathcal{S}_d \\
&\Pr \left\{\mathop {\lim }\limits_{T \to \infty } \| \widehat{A}_d(T) - {A_d} \|_F^2 \!=\! \bm{O}(1) \right\}\!=\!1,~\text{otherwise}\\
\end{aligned}\right..
\end{equation}
\end{theorem}

\begin{proof}
First, we consider the constant cooperation pattern case when $A_d \in \mathcal{S}_d$. 
Note that the sample matrix $\Sigma_1(T)$ can be written as 
\begin{align}\label{eq:case1_expression}
\Sigma_1(T)=& \frac{1}{T} \sum\limits_{t = 1}^{T }\left(A_d y_{t-1} + \upsilon_{t}-A_d \upsilon_{t-1}\right)y_{t-1}^\intercal \nonumber \\
=& A_d \Sigma_0(T) \! + \! \frac{1}{T} ({\Phi_T^+} - A_d {\Phi_T^-}) ({Y_T^-})^{\intercal} \nonumber \\
=& A_d \Sigma_0(T) + \frac{ {\Phi_T^+} ({X_T^-})^{\intercal} }{T} + \frac{ {\Phi_T^+} ({\Phi_T^-})^{\intercal} }{T} \nonumber \\
&- \frac{ A_d {\Phi_T^-} ({X_T^-})^{\intercal} }{T} - \frac{ A_d {\Phi_T^-} ({\Phi_T^-})^{\intercal} }{T}.
\end{align}
Notice that the elements in $X_T^-$ are bounded when $A_d \in \mathcal{S}_d$, and ${\Phi_T^+}$ is independent of ${X_T^-}$ and ${\Phi_T^-}$. 
Then, when $T\to\infty$, we have with probability one that 
\begin{align}\label{eq:limit_expression}
\Sigma_1(\infty)= & A_d \Sigma_0(\infty) + 0 + 0 -0 -A_d (I_N \!\otimes\! \Gamma) \nonumber \\
\Rightarrow A_d = & \Sigma_1(\infty)(\Sigma_0(\infty) -I_N \!\otimes\! \Gamma)^{-1},
\end{align}
which proves the first conclusion.


Next, we focus on the case when $A_d \notin \mathcal{S}_d$. 
In this situation, the term $ \lim_{T\to\infty}{\Phi_T^+} ({X_T^-})^{\intercal} /{T}$ will not converge to zero, and thus \eqref{eq:case1_expression} no longer holds. 
To deal with this issue, we consider the singular value decomposition (SVD) of $Y_T^-$, 
\begin{equation}
Y_T^-=U \Lambda V^\intercal,
\end{equation}
where $U\in\mathbb{R}^{Nn \times Nn}$ is a unitary matrix and $V\in\mathbb{R}^{T \times Nn }$ a block of a unitary matrix, and $\Lambda\in\mathbb{R}^{Nn \times Nn}$ is a diagonal matrix consisting of the singular values of $Y_T^-$. 
Then, denote the inference error matrix as $E_{1}= \widehat{A}_d(T) - {A_d}$, and it satisfies
\begin{align}
E_{1}=& ({\Phi_T^+} - A_d {\Phi_T^-}) ({Y_T^-})^{\intercal} ( Y_T^- ({Y_T^-})^{\intercal} )^{-1} \nonumber \\
=&  ({\Phi_T^+} - A_d {\Phi_T^-}) V \Lambda U^\intercal (U \Lambda V^\intercal   V \Lambda U^\intercal)^{-1} \nonumber \\
=&  ({\Phi_T^+} - A_d {\Phi_T^-}) V \Lambda^{-1} U^\intercal  \nonumber \\
\Rightarrow ~ \| E_{1}\| \le & \|{\Phi_T^+} - A_d {\Phi_T^-}\| \cdot \| V\| \cdot \| \Lambda^{-1}\| \cdot \|U^\intercal\|  \nonumber \\
\le & \frac{\|\Phi_T^+\| + \|A_d\| \|{\Phi_T^-}\| }{\sigma_{\min}(Y_T^-)},
\end{align}
where $\sigma_{\min}(\cdot) $ is the minimum singular value of a matrix.  

Recalling $Y_T^-=X_T^- + \Phi_T^-$, since $X_T^-$ is deterministic and independent of $\Phi_T^-$, $Y_T^-$ can be regarded as a shifted version of the random matrix $\Phi_T^-$ with independent entries. 
Therefore, we can directly apply the Theorem 1 in \cite{tikhomirov2016smallest} to obtain a lower bound for $\sigma_{\min}(Y_T^-)$ as
\begin{equation}\label{eq:pr_bound1}
\Pr\{ \sigma_{\min}(Y_T^-)\ge u\sqrt{T} \}\ge 1- \exp(-vT),
\end{equation}
where $u$ and $v$ are positive constants that do not depend on $T$. 
Next, for the noise matrix $\Phi_T^+$ where the entries are i.i.d. zero-mean Gaussian noises, we can apply the concentration of measure in Gaussian space \cite{davidson2001local} to $\Phi_T^+$ and obtain that with probability at least $1-2 \exp \left(-r^{2} / 2\right)$ 
\begin{equation}
 \sigma_{\max}(\Phi_T^+) \leq  \sqrt{T}+\sqrt{Nn}+r. 
\end{equation}
By setting $r=\sqrt{T}$, we directly have 
\begin{equation}\label{eq:fenzi}
\Pr\{ \sigma_{\max}(\Phi_T^+) \leq 2 \sqrt{T}+\sqrt{Nn}  \}  \geq 1-2 \exp \left(-T / 2\right) .
\end{equation}
The above upper bound with probability guarantee also holds for $\sigma_{\max}(\Phi_T^-)$.

Finally, considering the union probability of \eqref{eq:pr_bound1} and \eqref{eq:fenzi}, we have with probability at least $(1-2 \exp (-T / 2))(1- \exp(-vT))$  that 
\begin{align}
\| E_{1}\| \le & \frac{ (1 + \|A_d\|)  \| \sigma_{\max}(\Phi_T^+) \|   }{  \sigma_{\min}(Y_T^-) }  \!\le\! \frac{ (1 \!+\! \|A_d\|) ( 2 \sqrt{T} \!+\!\sqrt{Nn} )  }{  u\sqrt{T} }  \nonumber \\ 
<&  \frac{ (1 + \|A_d\|) ( 3 \sqrt{T}  )  }{ u\sqrt{T} }  = \frac{ 3(1 + \|A_d\|)  }{ u }=\bm{O}(1). 
\end{align}
Specifically, when $T\to\infty$, the probability becomes one, which completes the proof. 
\end{proof}

Theorem \ref{th:matrix_error} shows that the asymptotic inference error of $\widehat{A}_d(T)$ is unbiased when $A_d \in \mathcal{S}_d$, while it is upper bounded by a constant when $A_d \notin \mathcal{S}_d$. 
The latter outcome is caused by the non-converging state, making $ \lim_{T\to\infty}{\Phi_T^+} ({X_T^-})^{\intercal} /{T} \neq 0$. 
Indeed, if one can find a way to process the observations and make the processed observations exhibit convergence, then it is possible to construct an unbiased estimator. 
Following this idea, we next show how to deal with a special case of $A_d \notin \mathcal{S}_d$, where the node states evolve in a linear growth pattern asymptotically, 
i.e., there exists a constant vector $\bm{c}\in\mathbb{R}^{n\times 1}$ ($\bm{c}\neq\bm{0}$) such that 
\begin{equation}\label{eq:linear_speed}
\lim_{t \to \infty } \| x_i(t)- \bm{c}t \|=0,~\forall i\in\mathcal{V}. 
\end{equation}
For the case \eqref{eq:linear_speed}, we introduce the following second-order difference based empirical error within a time window $T$
\begin{equation}\label{eq:criteria_second}
\epsilon_2= \sum\limits_{k=T-L}^{T-1} \frac{ \left \| y(k+1)-2y(k)+y(k-1) \right \|}{L}. 
\end{equation}
If $\epsilon_1>\epsilon$ and $\epsilon_2<\epsilon$ hold simultaneously, it is regarded that the NDS evolves in a linear growth. 
Then, similar to $\Sigma_0(T)$ and $\Sigma_1(T)$, define the following matrices 
\begin{align}
Y_{\Delta,T}^- \!=\! [\Delta y(0),\Delta y(1),\!\cdots\!,\Delta y(T-2)] \in \mathbb{R}^{Nn \times (T-1)}, \nonumber  \\
Y_{\Delta,T}^+ \!=\! [\Delta y(1),\Delta y(2),\!\cdots\!,\Delta y(T-1)]  \in \mathbb{R}^{Nn \times (T-1)},  \nonumber 
\end{align}
where $\Delta y(t)\!=\!y(t+1)- y(t), t=0,\cdots\!,T\!-\!1$ is the observation difference at instant $t$. 
Correspondingly, the covariance matrix for $\Delta y(t)$ and its one-lag version are given by
\begin{equation}\label{eq:delta_covariance}
\Sigma_{\Delta,0}(T)\!=\!\frac{1}{T}(Y_{\Delta,T}^-) (Y_{\Delta,T}^-)^\intercal,\Sigma_{\Delta,1}\!=\!\frac{1}{T}(Y_{\Delta,T}^+) (Y_{\Delta,T}^-)^\intercal.
\end{equation}
\begin{corollary}\label{lemma:linear_case}
For the linear growth case \eqref{eq:linear_speed}, 
define the estimator for $A_c$ as $\widehat{A}_d^s(T)=(\Sigma_{\Delta,1}(T)+ I_N \!\otimes\! \Gamma)(\Sigma_{\Delta,0}(T)-2 I_N \!\otimes\! \Gamma )^{-1}$, then, with probability one, we have
\begin{equation}\label{eq:inty_performance}
\mathop {\lim }\limits_{T \to \infty } \| \widehat{A}_d^s(T) - A_d \|=0.
\end{equation}
\end{corollary}

\begin{proof}
First, when no confusion is caused, we will alternatively use $\Delta y_{t}$ for $\Delta y(t)$, while the convergence stated in this proof refers to almost sure convergence. 
By the definition of $\Sigma_{\Delta,1}(T)$, it can be expanded as 
\begin{align}\label{eq:vv_expand0}
&\Sigma_{\Delta,1}(T)= \frac{1}{T} \sum_{t=0}^{T-1} {\Delta y_{t+1}} {\Delta y_{t}^\intercal} \nonumber \\
=& \frac{1}{T} \sum_{t=0}^{T-2} \left(  A_d {\Delta y_{t}} - A_d {\Delta\upsilon_t}  + {\Delta\upsilon_{t+1} }   \right) {(\Delta x_{t} + {\Delta\upsilon_t}  )^\intercal}  \nonumber \\
=& A_d \Sigma_{\Delta,0}(T+1) +  \sum_{t=0}^{T-2} \frac{  \left( {\Delta\upsilon_{t+1} }  - A_d {\Delta\upsilon_t}  \right) {(\Delta x_{t} + {\Delta\upsilon_t}  )^\intercal}  }{T} ,
\end{align}
where $\Delta\upsilon_t=\upsilon_{t+1} - \upsilon_{t}$ is the noise difference at instant $t$.  
Note that when the NDS is subject to the linear growth \eqref{eq:linear_speed}, the second term in \eqref{eq:vv_expand0} will converge as $T\to\infty$, given by 
\begin{align}\label{eq:vv_expand_part2}
&\lim_{T \to \infty }  \frac{  \sum_{t=0}^{T-2} \left( {\Delta\upsilon_{t+1} }  - A_d {\Delta\upsilon_t}  \right) {(\Delta x_{t} + {\Delta\upsilon_t}  )^\intercal}  }{T} \nonumber \\
=& \lim_{T \!\to\! \infty } \frac{  \sum\limits_{t=0}^{T-2} \!\! {\Delta\upsilon_{t+1}} {\Delta x_{t}^{\intercal}} \!+\! {\Delta\upsilon_{t+1}} {\Delta\upsilon_t^{\intercal}}  \!-\! A_d {\Delta\upsilon_t} {\Delta x_{t}^{\intercal}} \!-\! A_d {\Delta\upsilon_t}{\Delta\upsilon_t^{\intercal}}  }{T} \nonumber \\
=& 0+ \lim_{T \to \infty } \sum_{t=0}^{T-2}\frac{ (\upsilon_{t+2} -\upsilon_{t+1}) (\upsilon_{t+1} -\upsilon_{t})^{\intercal}  }{T}  \nonumber \\
&- 0 - \lim_{T \to \infty } \sum_{t=0}^{T-2}\frac{ A_d(\upsilon_{t+1} -\upsilon_{t}) (\upsilon_{t+1} -\upsilon_{t})^{\intercal}  }{T}   \nonumber \\
=&0 - I_N \!\otimes\! \Gamma - 0 -  A_d \cdot 2 (I_N \!\otimes\! \Gamma),
\end{align}
where $\lim\limits_{T \to \infty } \sum_{t=0}^{T-2} \frac{{\Delta\upsilon_{t+1}} {\Delta x_{t}^{\intercal}}}{T} = \lim\limits_{T \to \infty } \sum_{t=0}^{T-2} \frac{A_d {\Delta\upsilon_t} {\Delta x_{t}^{\intercal}}}{T}=0 $ holds because $\Delta x_{t}$ is bounded in the linear growth case. 
Therefore, substituting \eqref{eq:vv_expand_part2} into \eqref{eq:vv_expand0} with $T \!\to\! \infty $, we have 
\begin{align}\label{eq:vv_expand}
\!\!\!&\Sigma_{\Delta,1}(\infty)=  A_d \Sigma_{\Delta,0}(\infty) - I_N \!\otimes\! \Gamma  - A_d \cdot 2 I_N \!\otimes\! \Gamma \nonumber \\
\!\!\! \Rightarrow & \Sigma_{\Delta,1}(\infty) + I_N \!\otimes\! \Gamma =  A_d \left( \Sigma_{\Delta,0}(\infty) -  2 I_N \!\otimes\! \Gamma \right)  \nonumber \\
\!\!\! \Rightarrow& A_d \!=\! (\Sigma_{\Delta,1}(\infty) \!+\! I_N \!\otimes\! \Gamma)( \Sigma_{\Delta,0}(\infty) \!-\!  2 I_N \!\otimes\! \Gamma)^{-1} \!=\! \widehat{A}_d^s(\infty),
\end{align}
which further verifies \eqref{eq:inty_performance} and completes the proof. 
\end{proof}

Lemma \ref{lemma:linear_case} shows that if the state trajectory of the system exhibits a linear growth in the asymptotic sense, one can use the difference of two consecutive states to remove the linear growth and construct an unbiased estimator for the matrix $A_d$. 
More generally, if the state trajectory has a polynomial growth asymptotically, then one can first adopt the polynomial fitting to obtain the trajectory vector representation $f(k)$ for the state. 
Then, one further use the filtered observations $\{y(k)-f(k),k=0,\cdots,T-1\}$ to construct the following equations and infer $A_d$
\begin{equation}\label{eq:series_eq}
y(k)\!-\!f(k)=A_d(y(k-1) \!-\!f(k-1)),k=1,\cdots,T. 
\end{equation}
Notice that in this case the inference error will largely depend on the fitting error in the first step. 
If the polynomial fitting $f(t)$ is unbiased in the asymptotic sense, then the inferred $A_d$ by solving \eqref{eq:series_eq} is also unbiased. 


\subsection{Reconstruct the Continuous-time System Matrix}


After the global discrete system matrix $A_d$ is inferred, 
this part presents the conditions and the reconstruction method for the continuous-time system matrix $A_c$ from a given $A_d$. 

To begin with, it is straightforward for one to come up with using the first-order difference of two observations to approximate the original $A_c$, i.e., supposing 
\begin{equation}\label{eq:appro}
\!\frac{x(k\!+\!1)\!-\!x(k)}{\tau} \approx \dot{x}(k)  \!=\! A_c x(k) \Rightarrow x(k\!+\!1)\!\approx\!(\tau A_c \!+\!I ) x(k).
\end{equation}
Then, an estimate of $A_c$ by \eqref{eq:appro} is $(A_d\!-\!I)/\tau$. 
However, this approximation only approaches the real $A_c$ when $\tau\to0$. 
To overcome the above drawback, we resort to the tool of matrix logarithm to reconstruct $A_c$ more reliably and derive a sufficient condition for the sampling period, presented in the following result.

\begin{theorem}\label{th:reconstruction}
Given an accurate $A_d$ under the sampling period $\tau$, if $\tau$ and the underlying $A_c$ satisfy $\tau  < \ln2 /\| A_c\| $, then $A_c$ can be uniquely reconstructed by 
\begin{equation}\label{eq:log_series}
A_c=\frac{1}{\tau}\log(A_d)= \frac{1}{\tau}\sum_{k=1}^{\infty}(-1)^{k+1} \frac{(A_d-I)^k}{k}. 
\end{equation}
\end{theorem}

\begin{proof}
The conclusion of this theorem has been proved in our previous work \cite[Theorem 2]{LI20238381} using geometric methods. 
Here we provide a different algebraic based proof, which consists of two parts: i) bound the norm $\| A_d - I \|$ by one, and ii) demonstrate the uniqueness of $\frac{1}{\tau}\log(A_d)$. 

First, denote the $i$-th (complex) eigenvalue of $A_c$ by 
\begin{equation}
\lambda_i = s_i +\textbf{i}\cdot \omega_i ,
\end{equation}
where $\textbf{i}$ is the imaginary number unit, $ s_i $ and $\omega_i $ are the real and imaginary parts of $\lambda_i$, respectively. 
Then, the eigenvalue of $A_d$ is the exponential of $\lambda_i \tau$, given by 
\begin{equation}\label{eq:complex_def}
\tilde{\lambda}_i=e^{\lambda_i \tau}\!=\!e^{s_i \tau +\textbf{i}\cdot \omega_i \tau }=e^{s_i \tau}(\cos (\omega_i \tau)+\textbf{i}\sin(\omega_i \tau)) ,
\end{equation}
which will never locate at the zero point, and thus $A_d$ is always invertible regardless of the invertibility of $A_c$. 
Notice that when $\tau  < \ln2 /\| A_c\| $, $\|A_d - I\|$ satisfies
\begin{align}\label{eq:bound_one}
&\|A_d - I\|=  \|e^{A_c \tau} - I\| = \left \| \sum_{k=1}^{\infty} \frac{(A_c \tau)^k}{k!} \right \|  \nonumber \\
\le & \sum_{k=1}^{\infty} \frac{\|A_c \tau\|^k}{k!}= e^{\|A_c \tau\|}-1<e^{\ln2}-1=1.
\end{align}
Substituting \eqref{eq:bound_one} into $\|\log(A_d)\|\le \sum_{k=1}^{\infty} \frac{\|A_d-I\|^k}{k}$, it follows that $\|\log(A_d)\|$ will converge.

Next, we show that $\log(A_d)=A_c \tau$ holds. 
For simplicity, we suppose that $A_c$ is diagonalizable (when $A_c$ is not diagonalizable, it can be approximated by a sequence of diagonalizable matrices). 
Then, $A_c$ can be expressed as $A_c=C J C^{-1}=C \operatorname{diag}\{\lambda_1,\cdots,\lambda_{Nn}\}C^{-1}$, and it holds that 
\begin{align}
(A_d\!-\!I)^k \!=\! C \operatorname{diag}\{(\tilde{\lambda}_i \!-\! 1)^k,\cdots,(\tilde{\lambda}_{Nn} \!-\!1)^k\} C^{-1} .
\end{align}
Since $\|A_c \tau\| <\ln2$ indicates $|\lambda_i \tau|<\ln2$, it follows by holomorphicity that $ \log ( e^{\lambda_i \tau})= \lambda_i \tau$ exactly even if $\lambda_i$ is a complex number (see \cite[Lemma 2.6]{hall2015lie}). 
Finally, we have
\begin{align}
&\sum_{k=1}^{\infty}(-1)^{k+1} \frac{(A_d \!-\!I)^k}{k} = C \operatorname{diag}\{\log{\tilde{\lambda}_i},\cdots,\log{\tilde{\lambda}_{Nn} } \} C^{-1} \nonumber \\
=& C \operatorname{diag}\{\lambda_1 \tau,\cdots,\lambda_{Nn}\tau\}C^{-1}=A_c \tau,
\end{align}
where $\log{\tilde{\lambda}_i}=\sum_{k=1}^{\infty}(-1)^{k+1} \frac{(\tilde{\lambda}_i-1)^k}{k} $ is used in the first row. 
Therefore, $A_c$ can be reconstructed by $\frac{1}{\tau}\log(A_d)$. 
The proof is completed. 
\end{proof}

Theorem \ref{th:reconstruction} presents a sufficient condition of accurately reconstructing $A_c$ from $A_d$. 
The condition $\tau  < \ln2 /\| A_c\|$ is sufficient in the sense that it avoids the non-uniqueness of $\log(A_d)$, considering the eigenvalues of $A_d$ can differ by integer multiples of $2\pi i$. 
Note that this sampling period condition can be relaxed to $\|e^{A_c \tau} - I \|<1$, if the eigenvalues of $A_c$ are all real numbers (e.g., when $A_c$ is known to be symmetric). 
It is also worth mentioning that $\tau  < \ln2/\| A_c\|  $ is stronger than the critical sampling period in Shannon's sampling theorem, i.e., 
\begin{equation}
\frac{\ln2}{\| A_c\|} <  \frac{\pi}{\max\nolimits_{i} \left|\operatorname{Im}\{\lambda_{i}\left(A_{c}\right)\}\right|}. 
\end{equation}
The Shannon's sampling period ensures to extract the spectrum of $A_c$ instead of $A_c$ itself (e.g., see \cite{chen2000reconstruction,ding2009reconstruction}). 
Besides, in our considered situation, $A_d$ is obtained based on noisy observations and thus is with certain bias. 
If we directly try to reconstruct the $A_c$ by finding the spectrum of $A_d$, the reconstruction error is hard to be characterized. 
In the next part, we will present detailed estimator design for $A_c$ and analyze the reconstruction error. 

\subsection{Asymptotic Inference Performance and Improved Design Under Limited Observations}

Note that in the inference setting, only the estimate $\widehat{A}_d$ is available and thus ${A}_c$ is reconstructed by
\begin{equation}\label{eq:estimate_Ac}
\widehat{A}_c=\frac{1}{\tau}\log(\widehat{A}_d).
\end{equation}
The feasibility of \eqref{eq:estimate_Ac} is based on two conditions. 
First, $\widehat{A}_d$ needs to be invertible such that it can be represented by the exponential of a matrix \cite[Theorem 2.10]{hall2015lie}. 
As we have analyzed in Remark \ref{rema:invertibility}, both the matrices $Y_T^+$ and $Y_T^-$ have full rank. 
Consequently, it holds that the inverse of $\widehat{A}_d(T)$ exists as shown in the following expression
\begin{equation}\label{eq:estimator_inverse}
\widehat{A}_{d}^{-1} \!=\!\left\{
\begin{aligned}
& \left( \Sigma_0(T) \!-\!I_N \!\otimes\! \Gamma \right) \Sigma_1(T)^{-1},\text{if}~A_d \in \mathcal{S}_d,\\
& \Sigma_0(T) (\Sigma_1(T) )^{-1},\text{otherwise},
\end{aligned}\right.
\end{equation}
where $\Sigma_1(T)$ is invertible in each case. 
Second, the convergence of \eqref{eq:estimate_Ac} requires $\|\widehat{A}_d -I\| <1$. 
Based on Theorem \ref{th:matrix_error}, we point out when $A_d \in \mathcal{S}_d$, this norm condition holds in high probability with sufficient observations since the inference error will decay to zero. 
When $A_d \notin \mathcal{S}_d$, the inference error will also decay to zero if $\sigma_{\min}(Y^{-}_{T})$ grows faster than $\bm{O}(\sqrt{T})$.

Next, we present the asymptotic inference performance of $\widehat{A}_c(T)$. 
For simplicity of notation, we define 
\begin{equation*}
\begin{aligned}
&\tilde{A}_d=A_d-I,~E_d(T)=\widehat{A}_d(T) -{A}_d , \\
&\tilde{\delta}=\|\tilde{A}_d\|,~\delta_T=\|E_d(T)\|,
\end{aligned}
\end{equation*}
where $\tilde{\delta}<1$ holds due to \eqref{eq:constraint_psd}. 
Recall that if $A_d \in \mathcal{S}_d$, then the constructed $\widehat{A}_d(T) $ is asymptotically unbiased with probability one. 
Based on this property, it is further deduced that there always exists a $T_d>0$ such that the following inequality holds with high probability
\begin{equation}\label{eq:error_D}
\delta_T<1-\tilde{\delta},~ \forall T\ge T_d. 
\end{equation}
Since we focus on characterizing the inference performance of $\widehat{A}_c(T)$ as $T$ grows, we directly consider $T\ge T_d$ and demonstrate that how $\delta_T$ will influence $\|\widehat{A}_c(T) -A_c\|$. 

\begin{theorem}\label{th:error_Ac}
Consider $\tau  < \ln2 /\| A_c\| $ and  $\widehat{A}_d(T)$ is obtained by \eqref{eq:discrete_estimator}. 
If $\widehat{A}_d(T)$ is asymptotically unbiased for $A_d\in\mathcal{S}_d$, 
then $\forall T\ge T_d$, the inference error of $\widehat{A}_c(T)$ is upper bounded, given by
\begin{equation}
\|\widehat{A}_c(T) -A_c\|\le \frac{ \delta_T }{\tau (1-\tilde{\delta}-\delta_T)},~\forall T\ge T_d. 
\end{equation}
Specifically, when $T\to \infty$, with probability one it holds that 
\begin{equation}
\lim_{T\to\infty} \|\widehat{A}_c(T) -A_c\| = 0 .
\end{equation}
\end{theorem}

\begin{proof}
The key of the proof is to bound the difference of the expansions of $\widehat{A}_c(T)$ and ${A}_c$. 


First, based on the matrix logarithm definition \eqref{eq:log_series}, the error matrix $E_c(T)=\widehat{A_c}(T)-A_c$ is written as 
\begin{align}\label{eq:Ac_difference}
E_c(T)&=\frac{1}{\tau}(\log(\widehat{A}_d(T))-\log(A_d)) \nonumber \\
&= \frac{1}{\tau}\sum_{k=1}^{\infty}(-1)^{k+1} \frac{ (\widehat{A}_d(T)-I)^k - (A_d-I)^k}{k} \nonumber \\
&= \frac{1}{\tau}\sum_{k=1}^{\infty}(-1)^{k+1} \frac{ (\tilde{A}_d+E_d(T) )^k - \tilde{A}_d^k}{k}. 
\end{align}
Since the commutativity between $\tilde{A}_{d}$ and $E_d(T)$ does not hold necessarily, the difference of $(\tilde{A}_{d} + E_d(T))^k-\tilde{A}_{d}^k $ can be generally expressed as 
\begin{equation}\label{eq:construction_expansion}
(\tilde{A}_{d} + E_d(T))^k-\tilde{A}_{d}^k= \operatorname{MP}_k(\tilde{A}_{d},E_d(T)) , 
\end{equation}
where $\operatorname{MP}_k(\tilde{A}_{d},E_d(T)) $ is matrix polynomial containing $(2^k-1)$ terms. 
Despite of the fussy expansion of $\operatorname{MP}_k(\tilde{A}_{d},E_d(T))$, its spectral norm can be upper bounded in an explicit form by utilizing $\| \tilde{A}_{d}^{\ell_1} E_d(T)^{\ell_2}\| \le \| \tilde{A}_{d}\|^{\ell_1} \|E_d(T)\|^{\ell_2} $ ($\forall \ell_1,\ell_2\ge0$), given by
\begin{align}\label{eq:Ad_norm}
& \|(\tilde{A}_{d} + E_d(T))^k-\tilde{A}_{d}^k\| = \| \operatorname{MP}(\tilde{A}_{d},E_d(T))  \| \nonumber \\ 
\le& \sum_{\ell=1}^{k} \! \binom{k}{\ell} \! \delta_T^{\ell} \tilde{\delta}^{k-\ell} \!=\! (\tilde{\delta}\!+\!\delta_T)^k\!-\!\tilde{\delta}^k = \delta_T \! \sum_{\ell=0}^{k-1} (\tilde{\delta} \!+\!\delta_T)^{k-\ell-1} \tilde{\delta}^{\ell} \nonumber \\
 \le &   \delta_T \! \sum_{\ell=0}^{k-1} (\tilde{\delta} \!+\!\delta_T)^{k-1} = \delta_T \cdot k (\tilde{\delta}+\delta_T)^{k-1},
\end{align}
where the first equality in the second row utilizes the binomial theorem and the second equality is to factor $(\tilde{\delta}+\delta_T)^k-\tilde{\delta}^k$. 

Next, substituting \eqref{eq:Ad_norm} into the norm of \eqref{eq:Ac_difference}, it yields that 
\begin{align}
\|E_c(T)\| &\le\! \sum_{k=1}^{\infty} \frac{ \| (\tilde{A}_d\!+\!E_d(T))^k \!- \!\tilde{A}_d^k  \| }{\tau k} \!\le\! \sum_{k=1}^{\infty} \frac{ (\tilde{\delta}+\delta_T)^k \!-\! \tilde{\delta}^k }{\tau k} \nonumber \\
& \le \frac{\delta_T}{\tau} \sum_{k=1}^{\infty} (\tilde{\delta}+\delta_T)^{k-1} = \frac{ \delta_T }{\tau (1-\tilde{\delta}-\delta_T)}. 
\end{align}
Finally, notice that when $A_d\in\mathcal{S}_d$, $\delta_T$ will converge to zero as $T\to\infty$ according to Theorem \ref{th:matrix_error}. 
Hence, it follows that 
\begin{equation}
\lim_{T\to\infty} \|E_c(T)\| \le \lim_{T\to\infty} \frac{ \delta_T }{\tau (1-\tilde{\delta}-\delta_T)}=0,
\end{equation}
which completes the proof. 
\end{proof}

Theorem \ref{th:error_Ac} illustrates that an unbiased $\widehat{A}_d(T)$ also ensures an unbiased $\widehat{A}_c(T)$ in the asymptotic sense. 
This conclusion also applies to the estimator $\widehat{A}_d^s(T)$ for the linear growth case, and the proof resembles that of Theorem \ref{th:error_Ac} and is omitted here. 
As for cases when $\{y_t\} $ grows in an ultra-linear pattern, the inference error of $\widehat{A}_c(T)$ is upper bounded by a constant. 
Specifically, the error will also decay to zero asymptotically if $\sigma_{\min}(Y^{-}_{T})$ grows faster than $\bm{O}(\sqrt{T})$.

Note that in some practical scenarios, the observations are not sufficiently large, and the constraint $\|\widehat{A}_d -I\| < 1$ may not hold. 
To avoid the occurrence of $\|\widehat{A}_d -I\| \ge 1$ under limited observations, we can formulate a semidefinite program optimization problem with constraints to obtain a feasible $\widehat{A}_d$. 
Notice that by using the Schur complement, the constraint $\|{A}_d - I \|<1$ is equivalent to the following linear matrix inequality 
\begin{equation}\label{eq:constraint_psd}
\begin{bmatrix} 
I & A_d-I \\
(A_d-I)^\intercal &  I
\end{bmatrix} \succ 0.
\end{equation}

Besides, it is worth noting that there also exists an implicit constraint for $A_d$. 
In the continuous model \eqref{eq:global_a}, the closed-loop system matrix $A_c$ satisfies
\begin{align}\label{eq:property_Ac}
A_c (\bm{1}_N \otimes I_n) & = (I_{N} \otimes A- L\otimes {B} K ) (\bm{1}_N \otimes I_n) = \bm{1}_N \otimes A \nonumber \\
& \Rightarrow A_c^k (\bm{1}_N \otimes I_n) = \bm{1}_N \otimes A^k,
\end{align}
where the fact $L \bm{1}_N = \bm{0}_N$ is used. 
Then, multiplying $A_d=e^{ {A_c} {\tau}}$ with $(\bm{1}_N \otimes I_n) $, it follows that 
\begin{align} \label{eq:property_Ad}
A_d (\bm{1}_N \otimes I_n) & =e^{ {A_c} {\tau}} (\bm{1}_N \otimes I_n)= \sum_{k=0}^{\infty} \frac{({A_c} {\tau}  )^k}{k!} (\bm{1}_N \otimes I_n)  \nonumber \\
&= \sum_{k=0}^{\infty} \frac{\bm{1}_N \otimes (A {\tau}  )^k}{k!} = \bm{1}_N \otimes e^{ {A} {\tau}}, 
\end{align}
where the property \eqref{eq:property_Ac} is applied in the second row. 
Although \eqref{eq:property_Ad} implies an equality constraint for $A_d$, the unknown matrix $A$ prohibits us from directly using it. 
However, if we treat $A_d (\bm{1}_N \otimes I_n)$ as a block matrix composed of $N$ identical submatrices in $\mathbb{R}^{n\times n}$, then we can construct the following equality from \eqref{eq:property_Ad} without using $A$, 
\begin{equation}\label{eq:constraint_eq}
A_d (\bm{1}_N \otimes I_n) - \bm{1}_N \otimes [A_d (\bm{1}_N \otimes I_n)]_{11}=\bm{0},
\end{equation}
where $[A_d (\bm{1}_N \otimes I_n)]_{11}$ represents the first $n\times n$ block in $A_d$ and its value should equal to $e^{ {A} {\tau}}$. 
Based on the conditions \eqref{eq:constraint_psd} and \eqref{eq:constraint_eq}, a feasible $\widehat{A}_d$ is obtained by solving the following problem for the case $\hat{\mathbb{I}}(A_d \in \mathcal{S}_d)=1$
\begin{subequations} \label{Prob:SDP}
\begin{align}
\mathop {\min }\limits_{A_d} ~& \| \Sigma_1(T) - A_d (\Sigma_0(T)-I_N \!\otimes\! \Gamma) \|_{F}^2  \\
\text{s.t.}~~& \eqref{eq:constraint_psd}~\text{and}~\eqref{eq:constraint_eq},
\end{align}
\end{subequations}
where the cost function is directly motivated from the estimator $\widehat{A}_d(T)=\Sigma_1(T) \left( \Sigma_0(T) \!-\!I_N \!\otimes\! \Gamma \right)^{-1}$. 
As for the case $\hat{\mathbb{I}}(A_d \in \mathcal{S}_d)=1$, we only need to change the objective function to $\| \Sigma_1(T) - A_d \Sigma_0(T)\|_{F}^2 $.

\subsection{Infer the Laplacian Topology Matrix}

In this part, we demonstrate how to extract an estimate for $L$ from the inferred continuous-time closed-loop matrix $\widehat{A}_c$.

First, recalling \eqref{eq:property_Ac} reveals that $A_c (\bm{1}_N \otimes I_n) = \bm{1}_N \otimes A$, which indicates that all $N$ blocks of the sum $A_c (\bm{1}_N \otimes I_n)$ are exactly the matrix $A$. 
This identity can be seen as an intrinsic disclosure vulnerability of NDSs satisfying \eqref{eq:global_a}. 
Leveraging this property, an external observer easily extracts the estimate of $A$ from $\widehat{A}_c$, given by 
\begin{equation}\label{eq:estimate_A2}
\widehat{A}=\frac{1}{N}\sum\limits_{i\in\mathcal{V}}\sum\limits_{j\in\mathcal{V}} \widehat{A}_c[\bm{i},\bm{j}]. 
\end{equation}
Since $\widehat{A}$ is a sum of finite blocks in $\widehat{A}_c$, the convergence rate of $\| \widehat{A}(T)-A\|$ is the same as $\|\widehat{A}_c(T)- A_c\|$ and it holds that $\lim_{T\to\infty} \|\widehat{A}(T) -A\| = 0$. 
Then, define an auxiliary matrix as 
\begin{equation}\label{eq:auxiliary_W}
W=I\otimes \widehat{A} - \widehat{A}_c.
\end{equation}
With the estimated $\widehat{A}_c$ and $\widehat{A}$, it is desired that the deviation $\| W - L\otimes {B} K\|$ should be as close to zero as possible. 
Since $BK$ is tied up with $L$ through a Kronecker product, it implies that $W$ contains multiple block-wise linear equations about $BK$ and $L$. 
Specifically, based on the definition of $L$, we have
\begin{equation}
\sum_{i,j=1, j\neq i } L_{ij} {B} K /N = -s_L BK,
\end{equation}
where $s_L= \frac{1}{N}\sum_{i=1} L_{ii}$ represents the average row-sum for $A$. 
Utilizing this property, a scalar ambiguity estimate of $BK$ is designed as 
\begin{align}\label{eq:solution_Q}
Z= \frac{1}{2N}{ \sum\limits_{i = 1}^{N} \left( W_{ii} - \sum\limits_{j = 1, j\neq i}^{N} W_{ij} \right)  },
\end{align}
which leverages all information contained in the matrix $W$. 
It is worth mentioning that $Z$ will approximate $s_L BK$ instead of $BK$ as the observation number goes to infinity, and thus we call it a scalar ambiguity estimate.

Next, we estimate the elements of $L$ consequently by taking the elements of $Z$ as benchmarks. 
Note that either $B_{\ell_1}=0$ or $K_{\ell_2}=0$ will result in $[BK]_{\ell_1 \ell_2}=0$, and thus the near-zero elements in $Z$ can be excluded for estimation. 
Define the index set of the non-zero elements of $Z$ as 
\begin{equation}\label{eq:Iz}
\mathcal{I}_Z=\{(\ell_1,\ell_2):~|Z_{\ell_1 \ell_2 }|>\varepsilon_Z, \ell_1,\ell_2=1,\cdots,n \},
\end{equation}
where $\varepsilon_Z>0$ is a small threshold to exclude the near-zero elements in $Z$. 
Then, a temporary surrogate for the element $L_{ij}$ can be computed by 
\begin{equation}\label{eq:estimate_L_other}
\tilde{L}_{i j} \!=\!  \frac{ 1 }{ |\mathcal{I}_Z| } \! \sum\limits_{ (\ell_1,\ell_2)\in\mathcal{I}_Z} \!\! \frac{ W[(i-1) n+{\ell_1},(j-1) n+{\ell_1} ]}{ Z_{\ell_1 \ell_2} }, 
\end{equation}
where $|\mathcal{I}_Z|$ represents the cardinality of $\mathcal{I}_Z$. 
Notice that the surrogate $\tilde{L}=[\tilde{L}_{i j}]_{i,j=1}^{N}$ is only $L$'s rough extraction from $\widehat{A}_c$, and does not necessarily satisfy the properties of a Laplacian matrix. 
Hence, we need to further infer $L$ by considering its properties and conditions that need to be met. 
This step is formulated as solving the following optimization problem
\begin{subequations} \label{Prob:infer_L}
\begin{align}
\mathop {\min }\limits_{L} ~& \| L - \tilde{L} \|_{F}^2  \\
\text{s.t.}~~& L\bm{1}=\bm{0},\\
&L_{ij}\le0,i\neq j, \\
&L~\text{has a diagonal Jordan form} . \label{Prob:infer_L4}
\end{align}
\end{subequations}
The constraint \eqref{Prob:infer_L4} is motivated by considering the there exists an optimal global LQ control objective function for the model \eqref{eq:node_dynamics}, which will be further inferred in Section \ref{sec:infer_objective}. 
If we only aim to infer the nodal dynamics and the topology matrix, this constraint can be dropped.

\begin{algorithm}[t]
\caption{Alternating minimization based topology inference algorithm}
\begin{algorithmic}[1]\label{algo:solve_L}
\REQUIRE Surrogate matrix $\tilde{L} \in \mathbb{R}^{ N\times N}$, stepsize factor $\alpha$, inner iteration limit $\bar{k}$, outer iteration limit $\bar{t}$, perturbation magnitude $\epsilon$, threshold parameter $\epsilon_L$
\ENSURE Estimate of matrix $L$;
\FOR {$t = 1, \ldots, \bar{t}$} 
{
    \STATE Initialize $\tilde{L}  = \tilde{L}^{(t-1)}$;
    \FOR{$k = 1, \ldots, \bar{k}$} 
    {
        \STATE Compute gradient: $\nabla f(L) = 2(L - \tilde{L} )$;
        \STATE Gradient descent update: $L \gets L - \alpha \nabla f(L)$;
        \STATE Project $L$ to satisfy row-sum constraint: $L \gets L - \frac{\mathbf{1} L \mathbf{1}^T}{N}$;
        \STATE Project $L$ to satisfy non-negativity constraint on off-diagonal: 
        \FORALL{$i \neq j$}
        {
            \STATE $L_{ij} \gets \max(L_{ij}, 0)$;
        }
        \ENDFOR
    }
    \ENDFOR
    \IF{$\|L-L^{(t-1)}\|\le \epsilon_L$ and $L$ is simple}
    {
        \STATE \textbf{break};
    }
    \ENDIF
    \STATE Add random perturbation: $L \gets L + \epsilon \cdot \text{randn}(N, N)$;
    \STATE Update $L^{(t)} \gets L$;
}
\ENDFOR
\STATE \textbf{Return} $L^{\bar{t}}$
\end{algorithmic}
\end{algorithm}

The problem \eqref{Prob:infer_L} is non-convex due to the diagonal Jordan form constraint, and conventional convex optimization methods may not work well. 
To deal with this issue, we adopt an alternating minimization based topology inference algorithm to approximate $L$, shown as Algorithm \ref{algo:solve_L}. 
This algorithm involves two major stages. 
In the first stage [Lines 2-11], without considering the diagonal Jordan form constraint, a projected gradient descent step is used to minimize the Frobenius norm $\|L-\tilde{L}\|_F$ subject to the row-sum and non-negativity constraints. 
A large inner iteration limit $\bar{k}\in\mathbb{N}^+$ is used to terminate the the iteration process (e.g., set $\bar{k}=500$). 
In the second stage [Lines 12-16], after the latest $\hat{L}$ is obtained, it is perturbed slightly by random noises to promote distinct eigenvalues, where the perturbation magnitude is characterized by a small $\epsilon$ (e.g., set as $0.05$). 
A small outer iteration limit $\bar{t}$ (e.g., set as $10$) is used to avoid infinite loops. 
Note that although the inferred $\widehat{L}=[\widehat{L}_{i j}]_{i,j=1}^N$ is also a scalar ambiguity estimate for $L$, its binary attribute (whether $\widehat{L}_{i j}$ is zero) is sufficient to reveal the interconnection relations between two nodes. 
By implementing Algorithm \ref{algo:solve_L}, we can obtain a feasible solution to problem \eqref{Prob:infer_L} effectively.


\subsection{Infer the Feedback Gain Matrix}
In this part, we propose a decomposition method to extract the state feedback matrix $K$ from the obtained $\widehat{A}_c$. 
For simplicity, we first analyze the single input case (i.e., $B\in \mathbb{R}^{n\times 1}$), and then demonstrate how to extend to multi-input case. 


First, note that any arbitrary two rows in $BK$ are linearly dependent when $B\in \mathbb{R}^{n\times 1}$.  
Supposing the $(\ell_1,\ell_2)$-th element of $BK$ is non-zero (i.e., $[BK]_{\ell_1 \ell_2}=B_{\ell_1} K_{\ell_2} \neq 0$, it holds that 
\begin{equation}\label{eq:B_ratio}
\frac{ [BK]_{\ell'_1 \ell_2} }{[BK]_{\ell_1 \ell_2} } = \frac{ B_{\ell'_1} }{ B_{\ell_1} },~ \ell'_1=1,\cdots,n. 
\end{equation}
Hence, we can utilize $Z$ to obtain a scalar ambiguity estimate of $B$. 
To begin with, the element with maximum magnitude in $Z$ is selected as the denominator in \eqref{eq:B_ratio}, denoted as
\begin{equation}
Z_{\ell_d\ell'_d}\!=\!\arg\max_{Z_{\ell_1 \ell_2}}\{|Z_{\ell_1 \ell_2}|: (\ell_1,\ell_2)\in\mathcal{I}_Z \},
\end{equation}
and the index set that excludes near-zero elements in the $\ell_d$-th row of $Z$ is defined by 
\begin{equation}\label{eq:index_j}
\mathcal{I}_Z^d= \{ \ell_2 : |Z_{\ell_d \ell_2}|>\varepsilon_Z, \ell_2=1,\cdots,n\},
\end{equation}
where $\varepsilon_Z$ is the same one in \eqref{eq:Iz}. 
Note that there exists at least one element in $\mathcal{I}_Z^d$ to enable the feedback control, and thus $\mathcal{I}_Z^d\neq\emptyset$. 
Similar to calculating $\tilde{L}$, we determine the values of elements in $B$ by 
\begin{equation}\label{eq:estimate_B}
\widehat{B}_{\ell_1}=\frac{ 1 }{|\mathcal{I}_Z^d|} \sum\limits_{\ell_2 \in\mathcal{I}_Z^d} \frac{ Z_{\ell_1 \ell_2} }{ Z_{ \ell_d \ell_2} },~ \ell_1=1,\cdots,n. 
\end{equation}
Next, considering that $B$ has full rank, the feedback gain matrix can be uniquely estimated by 
\begin{equation}\label{estimate:K}
\widehat{K}= ({\widehat{B}^\intercal} \widehat{B})^{-1} {\widehat{B}^\intercal} Z. 
\end{equation}

Based on the above procedures, we observe that the key of decoupling $K$ from the system topology builds on using the zero sum property of $L$ in the construction of $\widehat{Z}$. 
In addition, the accuracy of the previously estimated $\widehat{A}_c$ is fundamental for $\widehat{K}$. 
The following result demonstrates the asymptotic inference performance of $\widehat{K}$.


\begin{theorem}\label{th:single_input}
Supposing that $\widehat{A}_c$ is an asymptotically unbiased estimator, and $\widehat{L}$ and $\widehat{B}$ are determined by \eqref{eq:estimate_L_other} and \eqref{eq:estimate_B}, respectively, then $\widehat{K}$ approximates $K$ asymptotically up to a ambiguity scalar, i.e., 
\begin{equation}\label{estimate:K_new}
\mathop{\lim}\limits_{T \to \infty} \widehat{K}=\beta K,
\end{equation}
where $\beta \!=\!  s_L B_{\ell_d} $ is the ambiguity scalar. 
\end{theorem}

\begin{proof}
First, for the surrogate $\tilde{L}_{ij}$, it can be equivalent represented as 
\begin{equation}
\tilde{L}_{ij} \!=\!  \frac{ 1 }{ |\mathcal{I}_Z| } \! \sum\limits_{ (\ell_1,\ell_2)\in\mathcal{I}_Z} \!\!\!\! \frac{ W^*[(i-1) n+{\ell_1},(j-1) n+{\ell_1} ] + e^a_{\ell_1 \ell_2}  }{ Z^*_{\ell_1 \ell_2} + e^b_{\ell_1 \ell_2 }},
\end{equation}
where $W^*$ and $Z^*$ are the groundtruth of $L\otimes {B} K$ and $s_L BK$, 
$e^a_{\ell_1 \ell_2} $ is the estimation error of the $(\ell_1,\ell_2)$-th element in the $(i,j)$-th block of $W^*$, 
and $e^b_{\ell_1 \ell_2} $ is the estimation error of the $(\ell_1,\ell_2)$-th element of $Z^*$.  
Note that when $\widehat{A}_c$ is an asymptotically unbiased estimator, the estimates $W=I\otimes \widehat{A} - \widehat{A}_c$ and $Z$ computed by \eqref{eq:solution_Q} are also asymptotically unbiased. 
Hence, we have 
\begin{equation}
\mathop{\lim}\limits_{T \to \infty} e^a_{\ell_1 \ell_2} =\mathop{\lim}\limits_{T \to \infty} e^b_{\ell_1 \ell_2} =0,
\end{equation}
which indicates that $\tilde{L}_{ij}$ is asymptotically unbiased, and so is $\widehat{B}$ by similar derivation steps. 
To ease expressions, we directly treat $Z$, $\widehat{L}$ and $\widehat{B}$ as accurate hereafter, i.e.,  
\begin{equation}
Z=s_L BK,~\widehat{L}=\tilde{L}=\frac{L}{s_L},~\widehat{B}=\frac{B}{B_{\ell_d}}.
\end{equation}
Then, by the construction of $\widehat{K}$, it follows that 
\begin{align}
\widehat{K}&=\widehat{B}^{\dagger} \widehat{B} Z =(\frac{B^\intercal B}{B_{\ell_d}^2})^{-1} \cdot \frac{ B^\intercal }{B_{\ell_d}} \cdot (s_L BK) \nonumber \\
&= (B_{\ell_d} s_L)  ({B^\intercal B})^{-1} ({ B^\intercal }  B )K = (B_{\ell_d} s_L) K, 
\end{align}
which completes the proof. 
\end{proof}


Theorem \ref{th:single_input} reveals that the feedback matrix $K$ in the single-input case can be inferred up to an ambiguity scalar, which involves two factors brought by estimating $L$ and $B$ independently. 
Note that calculating the ratio of elements in $K$ and $B$ is the key to this conclusion, which is directly feasible in the single input case. 

For situations where $B\in \mathbb{R}^{n\times m} $ ($m>1$), the matrices $B$ and $K$ cannot be estimated uniquely due to bilinearity of the equation $BK=Z$. 
To deal with this issue, we leverage the full rank constraint of $B$ and $K$ and proceed to approximate them by the SVD factorization. 
By decomposing the matrix $Z$ and carefully selecting components that ensure the desired rank properties, we can construct solutions for $B$ and $K$ that meet both the equation and the rank constraints.
First, applying the SVD on the target matrix $Z$, we have
\begin{equation}
Z = \tilde{U} \tilde{\Lambda} \tilde{V}^\intercal,
\end{equation}
where $\tilde{U} \in \mathbb{R}^{n \times n}$ and $\tilde{V} \in \mathbb{R}^{n \times n}$ are orthogonal matrices, and $\tilde{\Lambda}\in \mathbb{R}^{n \times n}$ is a diagonal matrix composed of the singular values of $Z$. 
Since the goal is for $B$ and $K$ to both have rank $m$, we can focus on the top $m$ singular values and their corresponding vectors in $\tilde{U}$ and $\tilde{V}$. 
Let $\tilde{U}_m$ denote the first $n \times m$ columns of $\tilde{U}$, $\tilde{V}_m$ denote the first $n \times m$ columns of $\tilde{V}$, and $\tilde{\Lambda}_m \in \mathbb{R}^{m \times m}$ be the diagonal matrix composed of the top $m$ singular values of $\tilde{\Lambda}$. Then, we approximate $Z$ by
\begin{equation}
Z \approx \tilde{U}_m \tilde{\Lambda}_m \tilde{V}_m^\intercal.
\end{equation}
Using this truncated SVD, we construct $B$ and $K$ as
\begin{equation}\label{eq:BK}
\widehat{B} = \tilde{U}_m \tilde{\Lambda}_m^{1/2} \quad \text{and} \quad \widehat{K} = \tilde{\Lambda}_m^{1/2} \tilde{V}_m^T,
\end{equation}
where $\tilde{\Lambda}_m^{1/2}$ is the square root of the diagonal matrix $\tilde{\Lambda}_m$. 
This construction ensures that $B K = \tilde{U}_m \tilde{\Lambda}_m \tilde{V}_m^T$, which approximates $Z$ closely while satisfying the desired rank constraints. 
Additionally, the product $\widehat{K} \widehat{B} = \tilde{\Lambda}_m$ is an $m \times m$ matrix with full rank $m$, thus meeting the requirement that $K B$ is invertible. 
Consequently, this approach provides a solution that approximates $Z$ in a least-squares sense while maintaining the required rank properties.



\section{Inferring the LQ Cost Function}\label{sec:infer_objective}

Considering the feedback control of the NDS is governed by a standard LQ cost function, 
this section further demonstrates how to reconstruct the cost parameters therein.



\subsection{Conditions of Optimal LQ Control for NDSs}

First, recall that the nodal dynamics model is described by \eqref{eq:node_dynamics}. 
By treating all nodes as a whole system with its own dynamics, the global model can be written
\begin{equation}
\left\{
\begin{aligned}\label{eq:global_model}
\dot{x}&= (I_{N} \otimes A) x +(I_{N} \otimes B) u \triangleq \tilde{A} x + \tilde{B} u \\
u&=-(L\otimes K) x \triangleq -\tilde{K} x
\end{aligned}
\right.,
\end{equation}
where $\tilde{A}$ and $\tilde{B}$ are the global system matrices, and $\tilde{K}$ is the global state feedback gain matrix. 
From the LQ control perspective, this state feedback control can be formulated as minimizing the following cost
\begin{equation}\label{eq:objective}
J=\int_0^{\infty}\left(x^{\intercal} Q x+u^{\intercal} R u\right) d t,
\end{equation}
where $Q\in \mathbb{R}^{Nn \times Nn}$ is a positive semi-definite (PSD) matrix, and $R\in \mathbb{R}^{Nm \times Nm}$ is a positive definite (PD) matrix. 
Correspondingly, the optimal control law of \eqref{eq:objective} is given by
\begin{equation}\label{eq:k-LQR}
u=-\tilde{K}x=-R^{-1} \tilde{B}^{\intercal} Px ,
\end{equation}
where $P \in\mathbb{R}^{Nn\times Nn}$ is the unique solution of the following algebraic Riccatti equation (ARE) 
\begin{equation}\label{eq:global_ARE}
\tilde{A}^{\intercal} P+ P \tilde{A}- P \tilde{B} R^{-1} \tilde{B}^{\intercal} P+Q=\bm{0}. 
\end{equation}
Note that the formulation \eqref{eq:k-LQR}-\eqref{eq:global_ARE} is a forward result of LQ control design given the cost function.
Many existing works have investigated the optimal cooperative control for NDSs, e.g., see \cite{movric2014cooperative,zhang2015distributed}. 
Since this point is not the focus of our work, here we directly consider that $\tilde{K}$ is optimal for \eqref{eq:objective}.

To inversely determine the cost matrices $Q$ and $R$ from the system dynamics matrices and the feedback gain, 
the problem will become extremely difficult because the parameters of the LQ cost highly depend on the coupled nodal dynamics and the topology. 
More specifically, it follows from \eqref{eq:global_model} and \eqref{eq:k-LQR} that the global feedback gain should simultaneously satisfy $(L\otimes K) =\tilde{K}=R^{-1} \tilde{B}^{\intercal} P$, which indicates that the cost matrices have certain structure. 
We next show some properties that the cost matrices should satisfy due to the coupling nature. 


\begin{lemma}\label{lemma:QP}
Considering that \eqref{eq:k-LQR} is the optimal control for the LQ problem \eqref{eq:objective}, it is necessary that the matrices $Q$ and $P$ should satisfy
\begin{equation}\label{eq:necessary}
\sum_{j=1}^{N} Q_{ij}=\bm{0}, ~\sum_{j=1}^{N} P_{ij}=\bm{0},~ \forall i\in\mathcal{V}.
\end{equation} 
\end{lemma}

\begin{proof}
The proof is based on the the requirement of consensus (synchronization) on the NDS when the time goes to infinity. 
Let $x_c\in\mathbb{R}^{n}$ represent the synchronized state vector for all nodes. 
Notice that if the consensus is achieved, the global control input $u$ is zero regardless the value of $x_c$, i.e., 
\begin{equation}
u=\tilde{K} (\bm{1}_N \otimes x_c ) = (L \bm{1}_N \otimes {B} K x_c) = \bm{0}. 
\end{equation}
Based on this fact, multiplying the global ARE \eqref{eq:global_ARE} with $(\bm{1}_N \otimes x_c )$ yields that 
\begin{align}\label{eq:ARE_sum}
&(\tilde{A}^{\intercal} P+ P \tilde{A}- P \tilde{B} R^{-1} \tilde{B}^{\intercal} P+Q) (\bm{1}_N \otimes x_c ) \nonumber \\
=&(\tilde{A}^{\intercal} P + P \tilde{A} +Q )(\bm{1}_N \otimes x_c ) =\bm{0} \nonumber \\
\Leftrightarrow  & \left( \sum_{j=1}^{N}  a_{ji} \sum_{\ell=1}^{N} P_{j\ell}  +  \sum_{j=1}^{N} P_{ij} A  + \sum_{j=1}^{N} Q_{ij} \right ) x_c =\bm{0}, \forall i\in\mathcal{V} \nonumber \\
\Rightarrow  & \sum_{j=1}^{N} \left(  a_{ji} \sum_{\ell=1}^{N} P_{j\ell} \right) + \left( \sum_{j=1}^{N} P_{ij} \right ) A  + \sum_{j=1}^{N} Q_{ij}=\bm{0},  \forall i\in\mathcal{V} .  
\end{align}
where $P \tilde{B} R^{-1} \tilde{B}^{\intercal} P = P \tilde{B} \tilde{K} $ is applied in the first equality, and the last row is derived due to the arbitrariness of $x_c$. 
Next, note that the condition \eqref{eq:ARE_sum} should also hold regardless of the value of $A$. 
Therefore, when $A=\bm{0}$, we obtain  
\begin{equation}\label{eq:Q_sum}
\sum_{j=1}^{N} Q_{ij}=\bm{0},  \forall i\in\mathcal{V}.
\end{equation}
Then, substituting \eqref{eq:Q_sum} into $\eqref{eq:ARE_sum}$ yields that 
\begin{equation}
\sum_{j=1}^{N} \left(  a_{ji} \sum_{\ell=1}^{N} P_{j\ell} \right) + \left( \sum_{j=1}^{N} P_{ij} \right ) A=\bm{0} \Rightarrow \sum_{j=1}^{N} P_{ij}=\bm{0},
\end{equation}
where the arbitrariness of all columns in $A$ is used for the derivation. 
The proof is completed.  
\end{proof}

\subsection{Cost Matrices Inference}

In this part, we show how the cost matrices $\{Q,R\}$ can be obtained by leveraging the ARE \eqref{eq:global_ARE} and the condition \eqref{eq:necessary}. 

Note that the ARE \eqref{eq:global_ARE} does not involve $R$, and a straightforward method is to replace $P$ with an expression of $R$. 
Define auxiliary matrices $S\in\mathbb{R}^{N^2 n^2 \times N^2 n^2} $ and $H \!\in\!  \mathbb{R}^{N^2 n^2 \times 2 N^2 n^2}$ by 
\begin{align}
S &=I_{Nn} \otimes \tilde{A}^\intercal + \tilde{A}^\intercal \otimes I_{Nn}-  \tilde{K}^\intercal \tilde{B}^\intercal \otimes I_{Nn}, \label{eq:zii_def}\\
H &= [S , I_{N^2 n^2}]. 
\end{align}
Then, we characterize the solution space of $\{Q,P\}$. 
\begin{theorem}\label{th:numspace}
Let $\theta=[\operatorname{vec}^\intercal( P) ,\operatorname{vec}^\intercal(Q)]^\intercal \!\in\! \mathbb{R}^{2N^2 n^2}$. 
Then, we have 
\begin{equation}\label{eq:kernel_condition}
\theta \in \operatorname{ker}(H),
\end{equation}
where $\operatorname{ker}(H)$ is the kernel space of $H$ and always exists. 
\end{theorem}
\begin{proof}
The key of this proof is to utilize the dependence of $R$ on $P$. 
By substituting $\tilde{K}=R^{-1}\tilde{B}^\intercal P$ into the ARE \eqref{eq:global_ARE}, its equivalent vectorized form is written as  
\begin{align}\label{eq:QP_vector}
(& I_{Nn} \otimes A^\intercal) \operatorname{vec}( P) + (A^\intercal \otimes I_{Nn} ) \operatorname{vec}( P) \nonumber \\
&- (  \tilde{K}^\intercal \tilde{B}^\intercal \otimes I_{Nn} )\operatorname{vec}(P)+ \operatorname{vec}(Q)= \bm{0} \nonumber \\
&\Rightarrow S \operatorname{vec}(P) + \operatorname{vec}(Q)=\bm{0}. 
\end{align}
Clearly, $\{\widehat{P},\widehat{Q}\}$ can be directly obtained by solving the following homogeneous equations
\begin{equation}\label{eq:Hb}
H \theta=\bm{0},
\end{equation}
where all feasible solutions of $\theta$ in the above equation constitute the kernel of $H$, i.e., $\operatorname{ker}(H)$. 

Next, we show that $\operatorname{ker}(H)$ always exists. 
For the formula $H\theta=\bm{0}$, the number of equations is $N^2n^2$ while the number of unknowns is $2N^2n^2$. 
Meanwhile, the equality constraint \eqref{eq:necessary} needs to hold, which implies that there are $2Nn^2$ redundant variables in $\theta$. 
Hence, the number of nonredundant elements in $\theta$ is 
\begin{equation}
q=2N^2n^2-2Nn^2.
\end{equation}
Then, we have that $q>N^2n^2$ holds if and only if $N>2$. 
However, in this situation, the difference between $q$ and $\operatorname{rank}(W)$ satisfies
\begin{align}
q-\operatorname{rank}(W)\ge q- N^2 n^2\ge (N-2)Nn^2>3,
\end{align}
which is even looser than the necessary condition $q-\operatorname{rank}(W)=1$ to make $\theta$ a unique solution up to a scalar. 
Hence, $\operatorname{ker}(H)$ contains infinite solutions, and the proof is completed. 
\end{proof}

Theorem \ref{th:numspace} demonstrates the solution space of $\theta$, and provides a way to estimate the matrices $Q$ and $P$ by solving a group of homogeneous equations based on known $\{\tilde{A},\tilde{B},\tilde{K}\}$. 
If the solution $\theta$ is determined and $\widehat{P}$ is extracted accordingly, we can further leverage $\widehat{R} \tilde{K}=\tilde{B}^\intercal \widehat{P}$ and $\operatorname{rank}(\tilde{K})=\operatorname{rank}(\tilde{B})=\operatorname{rank}(\tilde{K}\tilde{B})=Nm$ to obtain 
\begin{equation}\label{eq:obtain_R}
 \widehat{R} \tilde{K} \tilde{B} = \tilde{B}^\intercal \widehat{P} \tilde{B} \Rightarrow  \widehat{R}= \tilde{B}^\intercal \widehat{P} \tilde{B} ( \tilde{K} \tilde{B})^{-1}.
\end{equation}

\begin{remark}
In \cite{molloy2022inversea}, the authors provide a way to directly obtain $\{\widehat{Q},\widehat{R}\}$ by using the inverse of $S$ and taking $\operatorname{vec}(P) =-S^{-1} \operatorname{vec}(Q)$ into the vectorized form of $R\tilde{K}=\tilde{B}^\intercal P$, given by
\begin{align}\label{eq:QR_vector2}
&(\tilde{K}^\intercal \otimes I_{Nm}) \operatorname{vec}(R)=( I_{Nn} \otimes \tilde{B}^\intercal) \operatorname{vec}(P) \nonumber \\
 \Rightarrow & ( I_{Nn} \otimes \tilde{B}^\intercal)S^{-1} \operatorname{vec}(Q)  \!  + \! (\tilde{K}^\intercal \otimes I_{Nm}) \operatorname{vec}(R) \!=\! \bm{0},
\end{align}
which implies that $\{Q,R\}$ can be obtained by solving a groups of homogeneous linear equations. 
However, in general, the invertibility of $S$ is not guaranteed, hindering the usage of \eqref{eq:QR_vector2} to obtain $\{Q,R\}$. 
In contrast, our method can avoid this issue.
\end{remark}

Note that in practice, we do not have the true values of $\{\tilde{A},\tilde{B},\tilde{K}\}$ but their corresponding estimates obtained in the last section. 
Accordingly, the estimate for $H$ is given by 
\begin{align}
\widehat{H}=[&I_{Nn} \otimes (I_{N} \otimes \widehat{A})^\intercal + (I_{N} \otimes \widehat{A})^\intercal \otimes I_{Nn}  \nonumber \\
 &-  (\widehat{L}\otimes \widehat{B}\widehat{K})^\intercal \otimes I_{Nn}, I_{N^2n^2}],
\end{align}
which is used to construct the equation $\widehat{H} \theta=\bm{0}$. 
Besides, it is worth mentioning that in many situations, the cost matrices $Q$ and $R$ are considered to be symmetric. 
When this symmetry is taken into account to solve $\widehat{H} \theta=\bm{0}$, the feasible solution will no longer exist because the number of non-redundant unknowns in $\theta$ becomes 
\begin{align}
\tilde{q} &=2N^2 n^2-2N n^2-(N^2n^2-Nn) \nonumber \\
&=N^2 n^2 + Nn - 2Nn^2 ,
\end{align}
which implies that it is smaller than the number of equations, i.e., $\tilde{q} < N^2 n^2$.

To make a solvable inverse optimal LQ control for NDSs, 
we treat the violation of $\widehat{H} \theta$ as an objective function and solve the following constrained  quadratic program problem
\begin{subequations}\label{eq:seek_theta}
\begin{align}
\mathop{\min }\limits_{\theta}~~& \theta^\intercal \widehat{H}^\intercal \widehat{H} \theta \label{eq:seek-a}\\
{\text{s.t.}}~~& Q\succ0, ~P\succ0, \label{eq:seek-b} \\
& \sum\nolimits_{j=1}^{N} Q_{ij}=\bm{0}, ~\sum\nolimits_{j=1}^{N} P_{ij}=\bm{0},~ \forall i\in\mathcal{V} , \label{eq:seek-d}\\
& Q=Q^\intercal ,~ P=P^\intercal . \label{eq:seek-c}
\end{align}
\end{subequations}
Notice that the constraint \eqref{eq:seek-b} enforces positive definiteness to avoid trivial solutions, and the constraint \eqref{eq:seek-c} can be dropped if the symmetry is not required. 
Since both the objective function and constraints in \eqref{eq:seek_theta} are convex, 
the optimal solution always exists (not necessarily unique), and one can resort to many mature optimization techniques to solve the problem. 
The matrix $R$ can be still estimated by \eqref{eq:obtain_R}.

In summary, the overall procedure of the proposed inverse inference method can be cast into three parts:
\begin{itemize}
\item Estimate the global closed-loop matrices: obtain the matrices $\widehat{A}_d$ by \eqref{Prob:SDP} and $\widehat{A}_c$ by \eqref{eq:estimate_Ac} consequently. 
\item Decouple the nodal dynamics and topology: estimate $A$ by \eqref{eq:estimate_A2}, $L$ by Algorithm \ref{algo:solve_L}, and $\{B,K\}$ by \eqref{eq:BK}.
\item Reconstruct the cost function: solve the parameters $\{Q,R\}$ by \eqref{eq:seek_theta} and \eqref{eq:obtain_R} consequently. 
\end{itemize}

\section{Simulations}\label{sec:simulation}
In this section, we provide numerical simulations to verify the theoretical results.

\begin{figure}[t]
\centering
\includegraphics[width=0.4\textwidth]{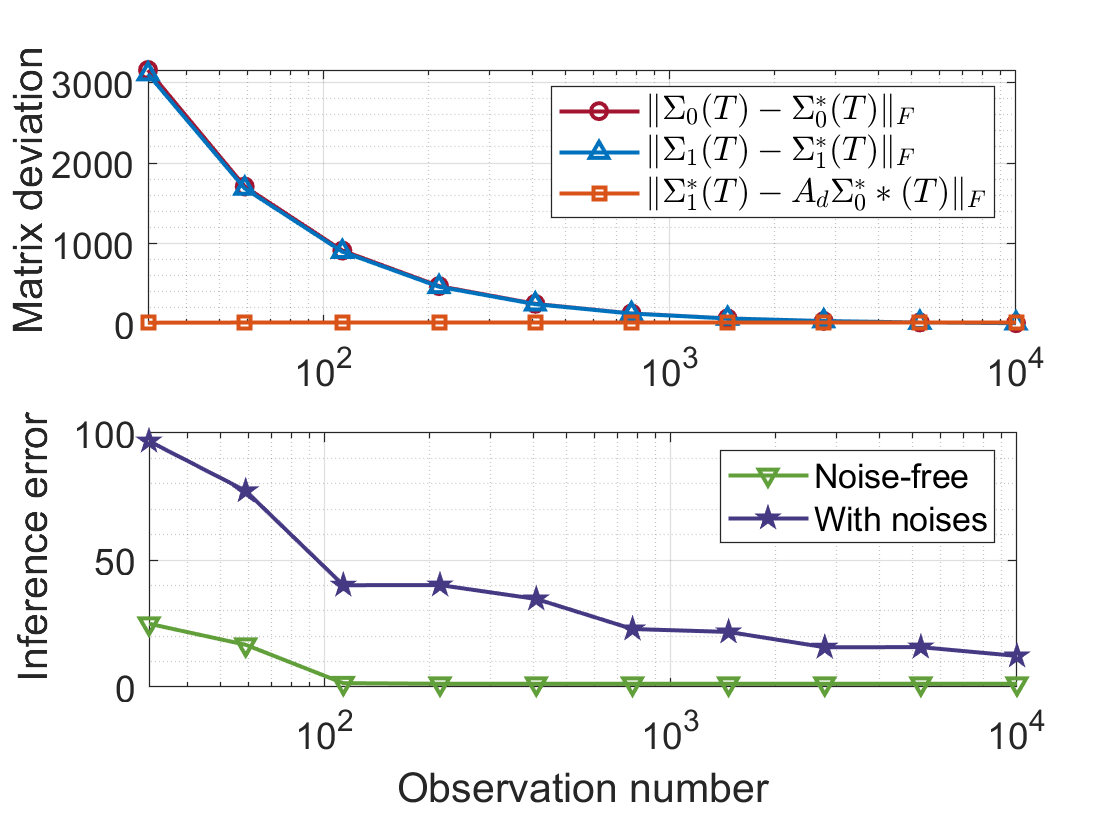}
\vspace{-10pt}
\caption{Asymptotic inference performance on $A_d$ by the naive estimator.}
\label{fig:asymptotic}
\end{figure}

\textit{Simulation setting}. We consider a NDS of $6$ nodes subject to the model \eqref{eq:node_dynamics}, where the adjacency topology $\textbf{A}_0$, nodal matrices $A$ and $B$, and the feedback gain $K$ are given by
\begin{align}
&\textbf{A}_0=\begin{bmatrix}
0& 1& 0& 0& 0& 0 \\
1& 0& 2& 0& 0& 0\\
0& 2& 0& 1& 0& 0\\
0& 0& 2& 0& 1& 0\\
0& 0& 0& 3& 0& 1\\
0& 0& 1& 0& 0& 0
\end{bmatrix},~
A=\begin{bmatrix}
-1& 2& 5\\ 
1& -1& 2\\
-5& 0& -1
\end{bmatrix}, \nonumber \\
&B^T=\begin{bmatrix}
0& 
0&
1
\end{bmatrix},~
K=\begin{bmatrix}
-0.0365  &  0.4295  &  0.9216
\end{bmatrix}. \nonumber 
\end{align}
It can be verified that the NDS will achieve the constant cooperation pattern in Definition \ref{def:constant}. 
The critical sampling period is $\ln{2}/\|A_c\|=0.073$s by Theorem \ref{th:reconstruction}. 
The initial states of each node in each dimension are randomly generated in $[0,1000]$. 
The state trajectory starts from $0$s and ends at $50$s. 
Five groups of noise standard deviation are given by $\text{G}_1=\operatorname{diag}\{2,1,0.2\}$, $\text{G}_2=\text{G}_1/2$, $\text{G}_3=\text{G}_1/4$, $\text{G}_4=\text{G}_1/20$, $\text{G}_5=\text{G}_1/40$. 
The inference error of an interested matrix $M_0$ is calculated by 
\begin{equation}
\text{Er}(\widehat{M_0},M_0)={\|\gamma \widehat{M_0}-M_0\|_{F}}/{\|M_0\|_{F}},
\end{equation}
where $\gamma$ represents the scalar ambiguity for $\widehat{M_0}$ (if there is).

\begin{figure}[t]
\centering
\includegraphics[width=0.4\textwidth]{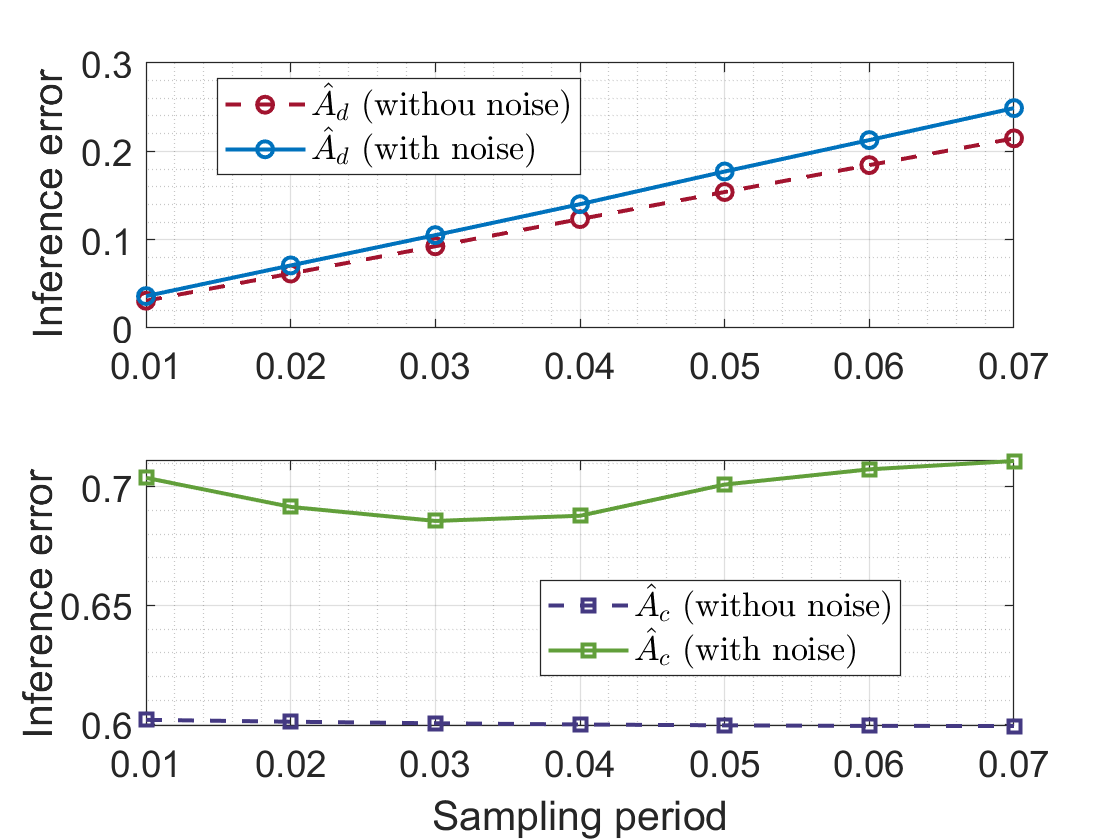}
\vspace{-8pt}
\caption{Performance under different sampling periods.}
\label{fig:sampling}
\end{figure}

\textit{Results analysis}. 
First, since the inference for $A_d$ is fundamental for the entire method, we present the asymptotic performance of the naive estimator $\widehat{A}_{d}(T)=\Sigma_1(T) \left( \Sigma_0(T) \!-\!I_N \!\otimes\! \Gamma \right)^{-1}$ in Fig.~\ref{fig:asymptotic}. 
In this test, the sampling period is set to $\tau=0.05$s and the noise standard deviation is specified as $\text{G}_4$. 
For clarity, the top subplot shows the deviation between the noisy sample matrix $\Sigma_0(T)$ and $\Sigma_1(T)$) versus their noise-free counterparts $\Sigma_0^*(T)=(X^- (X^-)^{\intercal})/T$ and $\Sigma_1^*(T)=(X^+ (X^-)^{\intercal})/T$, along with the residual $\|\Sigma_1^*(T)-A_d \Sigma_0^*(T)\|_{F}$. 
The plotted curves clearly indicate that both $\Sigma_0(T)$ and $\Sigma_1(T)$ will approach their noise-free counterparts asymptotically. 
Similarly, in the bottom subplot, the performance of the naive estimator also converges the noise-free version asymptotically, thereby confirming the conclusions of Theorem \ref{th:matrix_error}. 
However, it is important to note that the naive estimator does not achieve satisfactory accuracy when the number of observations is not sufficiently large. 
Hence, the naive estimator is considered an ideal design primarily in an asymptotic sense. 

In the following, we present the results of the proposed optimization based method. 
To begin with, we examine how the sampling period $\tau$ will influence the inference error of the discretized and continuous-time closed-loop matrices $A_d$ and $A_c$, respectively. 
In Fig.~\ref{fig:sampling}, results of both noise-free and noisy situations are plotted. 
In the noisy case, the noise standard deviation is set as $\text{G}_2$ and the curve value is computed by the average of $10$ random tests. 
The upper subfigure in Fig.~\ref{fig:sampling} depicts the errors of $\widehat{A}_d$, and it is evident that the inference error will increase as $\tau$ becomes larger. 
This trend aligns with our intuition, because given a fix time slot, a smaller sampling period brings more available observations, which is beneficial for inference. 
The bottom subfigure illustrates the error of $\widehat{A}_c$, which is based on the logarithm of $\widehat{A}_d$. 
In this test, the errors remain generally stable especially for the noise-free case. 
We observe that this is mainly because given finite observations, the error of $\widehat{A}_d$ has not exhibited asymptotic convergence yet. 
According to the construction formula $\frac{1}{\tau}\sum_{k=1}^{\infty}(-1)^{k+1} \frac{(\widehat{A}_d-I)^k}{k}$, both the denominator and the numerator terms will decrease as $\tau$ decreases, resulting in a nearly stable tendency in the plots. 

\begin{figure}[t]
\centering
\subfigure[Inference errors of $A_d$, $A_c$, $A$ and $BK$.]{\label{fig:error_observation}
\vspace{-10pt}
\includegraphics[width=0.4\textwidth]{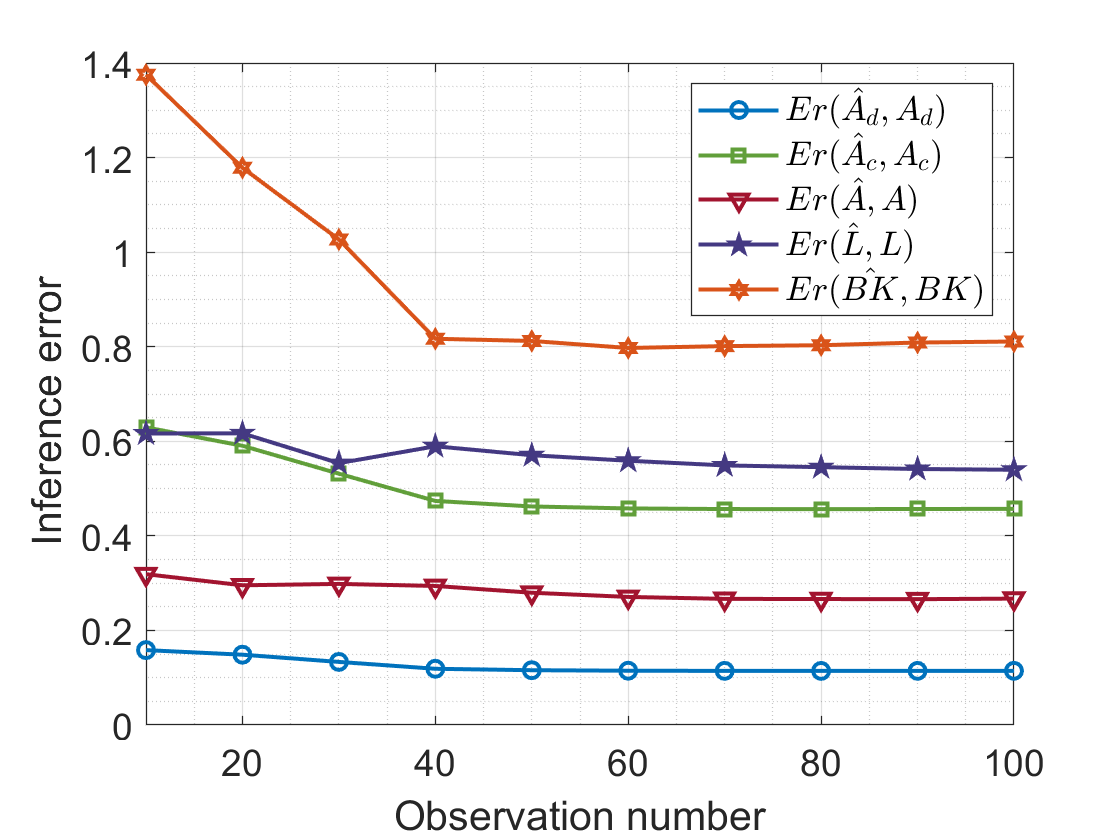}}
\subfigure[Inference errors of $(Q,P)$ and its trajectory error.]{\label{fig:error_LQ}
\includegraphics[width=0.4\textwidth]{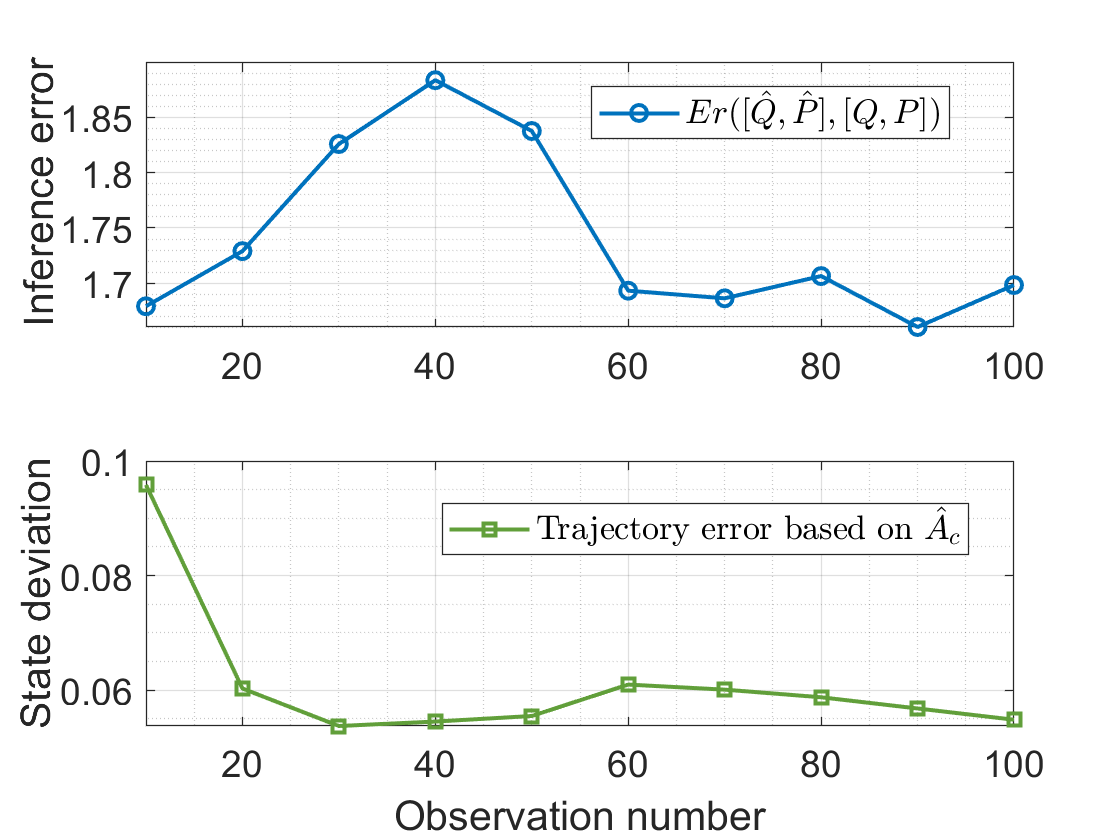}}
\vspace{-8pt}
\caption{Performance under $\tau=0.05$s and noise setting $G_2$.}
\label{fig:infer_parameters}
\end{figure}

Next, we provide the inference errors of all concerned parameters regarding the observation number with $\tau=0.05$. 
Fig.~\ref{fig:error_observation} presents the inference results of $A_d$, $A_c$, $A$ and $BK$\footnote{Since each node has multiple inputs ($m>1$) in this example, we directly use $BK$ to verify the accuracy up to a scalar ambiguity.} in a noisy setting, where the noise standard deviation is set as $\text{G}_2$. 
In these plots, the inference errors of the concerned matrices will generally decrease with $T$ initially, but remain slightly changed as $T$ becomes larger. 
Although this pattern may not correspond to our typical expectation that the errors should approach zero asymptotically, 
it is important to note that this behavior is largely attributable to numerical precision limitations inherent in computational processes (similar trends are noted even in noise-free scenarios). 
As the system state approaches zero in a short time, subsequent observations are almost zero-mean noises, 
contributing little to the accuracy of the estimator. 
In Fig.~\ref{fig:error_LQ}, the top subfigure displays the results of inferring the optimal control parameters $Q$ and $P$. 
Notably, the large inference error and fluctuations in curve $\text{Er}([\widehat{Q},\widehat{P}],[Q,P])$ primarily arise from the non-unique solutions to problem \eqref{eq:seek_theta}. 
Nevertheless, given the same initial state, these cost parameters are capable of generating trajectories closely resembling the actual ones, as shown in the bottom subfigure.

Finally, Fig.~\ref{fig:variance} depicts the inference errors of concerned matrices under five groups of noise standard deviation settings. 
In this test, we set the sampling period as $\tau=0.05$s and utilize $100$ observations for the estimators, with the curves representing the average values of 10 random implementations. 
It is clear to see from an arbitrary curve that, the inference error will decrease with the noise variance growing, 
which matches our intuitive expectation.

\begin{figure}[t]
\centering
\includegraphics[width=0.4\textwidth]{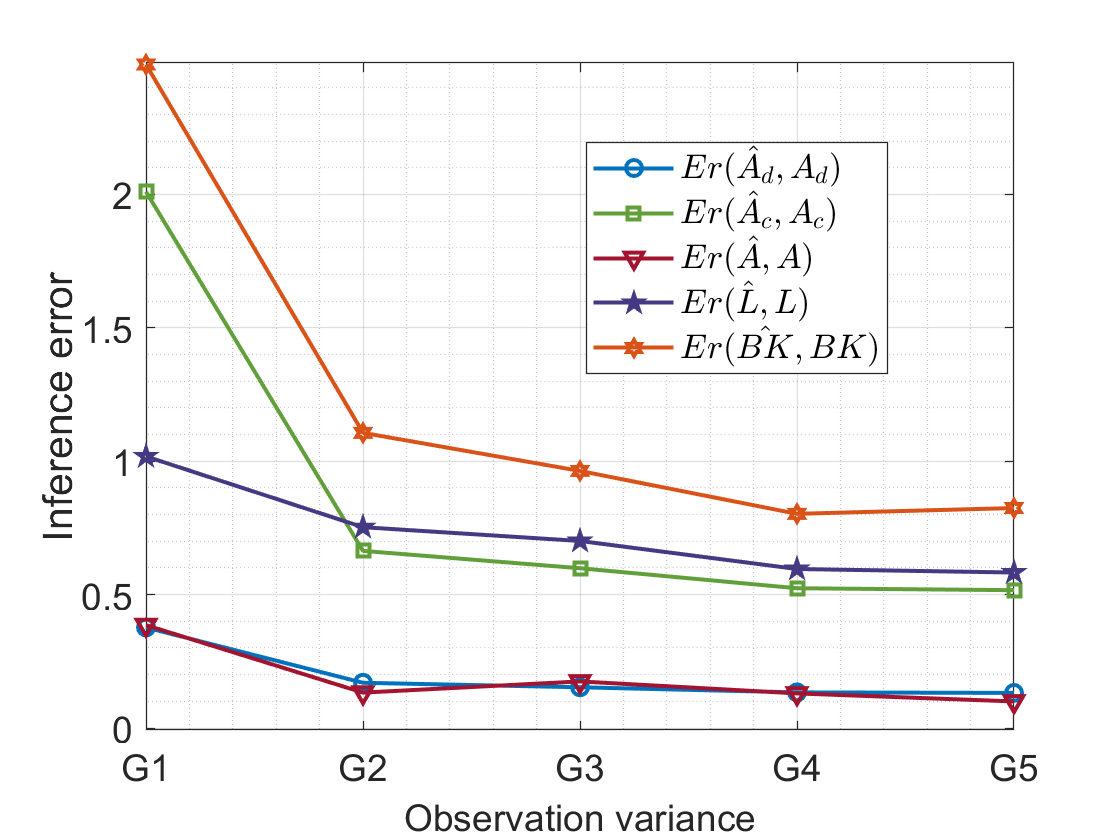}
\caption{Inference errors regarding different noise variances under $\tau=0.05$s.}
\label{fig:variance}
\end{figure}

\section{Conclusion}\label{sec:conclusions}

In this paper, we presented a bi-level framework for inferring the state feedback cooperative control of NDSs, where the interested components include the nodal dynamics, the topology between nodes, and the feedback control gain. 
Specifically, only noisy observations on a single trajectory of the NDS are available. 
At the first level, we utilized the causality of consecutive observations to construct the estimator for the discretized closed-loop system matrix, which was proved asymptotically unbiased when the NDS is stable. 
At the second level, we resorted to the matrix logarithm to recover the continuous-time system from its discrete-time counterpart, and derived the sampling period condition and the recovery error the bound. 
Following this, we developed least squares based procedures to decouple  the components. 
Furthermore, we demonstrated how to reconstruct the LQ objective function that governs the control, providing the necessary conditions of the corresponding parameters. 
Finally, numerical simulations illustrated the effectiveness of the proposed method. 

The investigation of this paper provides new insights to reveal the underlying mechanism of NDSs, and can pave the way for more promising directions, 
such as i) extending the method to a general heterogeneous NDS model, and ii) considering broader cooperative scenarios.


\begin{thebibliography}{10}
\providecommand{\url}[1]{#1}
\csname url@samestyle\endcsname
\providecommand{\newblock}{\relax}
\providecommand{\bibinfo}[2]{#2}
\providecommand{\BIBentrySTDinterwordspacing}{\spaceskip=0pt\relax}
\providecommand{\BIBentryALTinterwordstretchfactor}{4}
\providecommand{\BIBentryALTinterwordspacing}{\spaceskip=\fontdimen2\font plus
\BIBentryALTinterwordstretchfactor\fontdimen3\font minus
  \fontdimen4\font\relax}
\providecommand{\BIBforeignlanguage}[2]{{%
\expandafter\ifx\csname l@#1\endcsname\relax
\typeout{** WARNING: IEEEtran.bst: No hyphenation pattern has been}%
\typeout{** loaded for the language `#1'. Using the pattern for}%
\typeout{** the default language instead.}%
\else
\language=\csname l@#1\endcsname
\fi
#2}}
\providecommand{\BIBdecl}{\relax}
\BIBdecl

\bibitem{oh2015survey}
K.-K. Oh, M.-C. Park, and H.-S. Ahn, ``A survey of multi-agent formation
  control,'' \emph{Automatica}, vol.~53, pp. 424--440, 2015.

\bibitem{10113752}
D.~Deka, V.~Kekatos, and G.~Cavraro, ``Learning distribution grid topologies: A
  tutorial,'' \emph{IEEE Transactions on Smart Grid}, vol.~15, no.~1, pp.
  999--1013, 2024.

\bibitem{he2022uncovering}
X.~He, L.~Caciagli, L.~Parkes, J.~Stiso, T.~M. Karrer, J.~Z. Kim, Z.~Lu,
  T.~Menara, F.~Pasqualetti, M.~R. Sperling, J.~I. Tracy, and D.~S. Bassett,
  ``Uncovering the biological basis of control energy: Structural and metabolic
  correlates of energy inefficiency in temporal lobe epilepsy,'' \emph{Science
  Advances}, vol.~8, no.~45, p. eabn2293, 2022.

\bibitem{tajvar2023modelling}
P.~Tajvar, R.~Forlin, P.~Brodin, and D.~V. Dimarogonas, ``Modelling pathogen
  response of the human immune system in a reduced state space,'' in \emph{2023
  62nd IEEE Conference on Decision and Control ({{CDC}})}, 2023, pp. 715--720.

\bibitem{gao2023data}
T.-T. Gao and G.~Yan, ``Data-driven inference of complex system dynamics: A
  mini-review,'' \emph{Europhysics Letters}, vol. 142, no.~1, p. 11001, 2023.

\bibitem{vanwaarde2021topology}
H.~J. {van Waarde}, P.~Tesi, and M.~K. Camlibel, ``Topology identification of
  heterogeneous networks: Identifiability and reconstruction,''
  \emph{Automatica}, vol. 123, p. 109331, 2021.

\bibitem{LI20238381}
Y.~Li, T.~Xu, J.~He, C.~Chen, and X.~Guan, ``Inferring state-feedback
  cooperative control of networked dynamical systems,''
  \emph{IFAC-PapersOnLine}, vol.~56, no.~2, pp. 8381--8386, 2023, 22nd IFAC
  World Congress.

\bibitem{ljung1998system}
L.~Ljung, ``System identification,'' in \emph{Signal Analysis and
  Prediction}.\hskip 1em plus 0.5em minus 0.4em\relax Springer, 1998, pp.
  163--173.

\bibitem{simchowitz2018learning}
M.~Simchowitz, H.~Mania, S.~Tu, M.~I. Jordan, and B.~Recht, ``Learning without
  mixing: Towards a sharp analysis of linear system identification,'' in
  \emph{Proceedings of the 31st Conference On Learning Theory}.\hskip 1em plus
  0.5em minus 0.4em\relax {PMLR}, 2018, pp. 439--473.

\bibitem{sarkar2019near}
T.~Sarkar and A.~Rakhlin, ``Near optimal finite time identification of
  arbitrary linear dynamical systems,'' in \emph{International Conference on
  Machine Learning}.\hskip 1em plus 0.5em minus 0.4em\relax PMLR, 2019, pp.
  5610--5618.

\bibitem{zheng2021nonasymptotic}
Y.~Zheng and N.~Li, ``Non-asymptotic identification of linear dynamical systems
  using multiple trajectories,'' \emph{IEEE Control Systems Letters}, vol.~5,
  no.~5, pp. 1693--1698, 2021.

\bibitem{9946382}
Y.~Jedra and A.~Proutiere, ``Finite-time identification of linear systems:
  Fundamental limits and optimal algorithms,'' \emph{IEEE Transactions on
  Automatic Control}, vol.~68, no.~5, pp. 2805--2820, 2023.

\bibitem{10383489}
I.~Ziemann, A.~Tsiamis, B.~Lee, Y.~Jedra, N.~Matni, and G.~J. Pappas, ``A
  tutorial on the non-asymptotic theory of system identification,'' in
  \emph{2023 62nd IEEE Conference on Decision and Control (CDC)}, 2023, pp.
  8921--8939.

\bibitem{fattahi2020efficient}
S.~Fattahi, N.~Matni, and S.~Sojoudi, ``Efficient learning of distributed
  linear-quadratic control policies,'' \emph{SIAM Journal on Control and
  Optimization}, vol.~58, no.~5, pp. 2927--2951, 2020.

\bibitem{depersis2021lowcomplexity}
C.~De~Persis and P.~Tesi, ``Low-complexity learning of linear quadratic
  regulators from noisy data,'' \emph{Automatica}, vol. 128, p. 109548, 2021.

\bibitem{hu2023theoreticala}
B.~Hu, K.~Zhang, N.~Li, M.~Mesbahi, M.~Fazel, and T.~Ba{\c s}ar, ``Toward a
  theoretical foundation of policy optimization for learning control
  policies,'' \emph{Annual Review of Control, Robotics, and Autonomous
  Systems}, vol.~6, no.~1, pp. 123--158, 2023.

\bibitem{9903320}
W.~Liu, J.~Sun, G.~Wang, F.~Bullo, and J.~Chen, ``Data-driven resilient
  predictive control under denial-of-service,'' \emph{IEEE Transactions on
  Automatic Control}, vol.~68, no.~8, pp. 4722--4737, 2023.

\bibitem{yu2023onlineb}
J.~Yu, D.~Ho, and A.~Wierman, ``Online adversarial stabilization of unknown
  networked systems,'' \emph{Proceedings of the ACM on Measurement and Analysis
  of Computing Systems}, vol.~7, no.~1, pp. 1--43, 2023.

\bibitem{abazar2020inverse}
N.~Ab~Azar, A.~Shahmansoorian, and M.~Davoudi, ``From inverse optimal control
  to inverse reinforcement learning: {{A}} historical review,'' \emph{Annual
  Reviews in Control}, vol.~50, pp. 119--138, 2020.

\bibitem{Kalman1964optimal}
R.~E. Kalman, ``{When Is a Linear Control System Optimal?}'' \emph{Journal of
  Basic Engineering}, vol.~86, no.~1, pp. 51--60, 1964.

\bibitem{priess2015solutions}
M.~C. Priess, R.~Conway, J.~Choi, J.~M. Popovich, and C.~Radcliffe, ``Solutions
  to the inverse {LQR} problem with application to biological systems
  analysis,'' \emph{IEEE Transactions on Control Systems Technology}, vol.~23,
  no.~2, pp. 770--777, 2015.

\bibitem{menner2020maximum}
M.~Menner and M.~N. Zeilinger, ``Maximum likelihood methods for inverse
  learning of optimal controllers,'' \emph{IFAC-PapersOnLine}, vol.~53, no.~2,
  pp. 5266--5272, 2020.

\bibitem{jin2021inverse}
W.~Jin, D.~Kuli{\'c}, S.~Mou, and S.~Hirche, ``Inverse optimal control from
  incomplete trajectory observations,'' \emph{The International Journal of
  Robotics Research}, vol.~40, no. 6-7, pp. 848--865, 2021.

\bibitem{zhang2019inverseb}
H.~Zhang, J.~Umenberger, and X.~Hu, ``Inverse optimal control for discrete-time
  finite-horizon linear quadratic regulators,'' \emph{Automatica}, vol. 110, p.
  108593, 2019.

\bibitem{li2020continuoustime}
Y.~Li, Y.~Yao, and X.~Hu, ``Continuous-time inverse quadratic optimal control
  problem,'' \emph{Automatica}, vol. 117, p. 108977, 2020.

\bibitem{yu2021system}
C.~Yu, Y.~Li, H.~Fang, and J.~Chen, ``System identification approach for
  inverse optimal control of finite-horizon linear quadratic regulators,''
  \emph{Automatica}, vol. 129, p. 109636, 2021.

\bibitem{zhang2024statistically}
H.~Zhang and A.~Ringh, ``Statistically consistent inverse optimal control for
  discrete-time indefinite linear--quadratic systems,'' \emph{Automatica}, vol.
  166, p. 111705, 2024.

\bibitem{10697272}
C.~Qu, J.~He, and X.~Duan, ``Control law learning based on {LQR} reconstruction
  with inverse optimal control,'' \emph{IEEE Transactions on Automatic
  Control}, to appear, 2024.

\bibitem{mateos2019connecting}
G.~Mateos, S.~Segarra, A.~G. Marques, and A.~Ribeiro, ``Connecting the dots:
  Identifying network structure via graph signal processing,'' \emph{IEEE
  Signal Processing Magazine}, vol.~36, no.~3, pp. 16--43, 2019.

\bibitem{zaman2021online}
B.~Zaman, L.~M.~L. Ramos, D.~Romero, and B.~Beferull-Lozano, ``Online topology
  identification from vector autoregressive time series,'' \emph{IEEE
  Transactions on Signal Processing}, vol.~69, pp. 210--225, 2021.

\bibitem{hayden2016sparse}
D.~Hayden, Y.~H. Chang, J.~Goncalves, and C.~J. Tomlin, ``Sparse network
  identifiability via compressed sensing,'' \emph{Automatica}, vol.~68, pp.
  9--17, 2016.

\bibitem{bombois2023informativity}
X.~Bombois, K.~Colin, P.~M. {Van den Hof}, and H.~Hjalmarsson, ``On the
  informativity of direct identification experiments in dynamical networks,''
  \emph{Automatica}, vol. 148, p. 110742, 2023.

\bibitem{prasse2022predicting}
B.~Prasse and P.~Van~Mieghem, ``Predicting network dynamics without requiring
  the knowledge of the interaction graph,'' \emph{Proceedings of the National
  Academy of Sciences}, vol. 119, no.~44, p. e2205517119, 2022.

\bibitem{10337619}
E.~Restrepo, N.~Wang, and D.~V. Dimarogonas, ``Simultaneous topology
  identification and synchronization of directed dynamical networks,''
  \emph{IEEE Transactions on Control of Network Systems}, vol.~11, no.~3, pp.
  1491--1501, 2024.

\bibitem{wai2019joint}
H.-T. Wai, A.~Scaglione, B.~Barzel, and A.~Leshem, ``Joint network topology and
  dynamics recovery from perturbed stationary points,'' \emph{IEEE Transactions
  on Signal Processing}, vol.~67, no.~17, pp. 4582--4596, 2019.

\bibitem{10210488}
E.~Tan, D.~Corrêa, T.~Stemler, and M.~Small, ``A backpropagation algorithm for
  inferring disentagled nodal dynamics and connectivity structure of dynamical
  networks,'' \emph{IEEE Transactions on Network Science and Engineering},
  vol.~11, no.~1, pp. 613--624, 2024.

\bibitem{movric2014cooperative}
K.~H. Movric and F.~L. Lewis, ``Cooperative optimal control for multi-agent
  systems on directed graph topologies,'' \emph{IEEE Transactions on Automatic
  Control}, vol.~59, no.~3, pp. 769--774, 2014.

\bibitem{zhang2015distributed}
H.~Zhang, T.~Feng, G.-H. Yang, and H.~Liang, ``Distributed cooperative optimal
  control for multiagent systems on directed graphs: An inverse optimal
  approach,'' \emph{IEEE Transactions on Cybernetics}, vol.~45, no.~7, pp.
  1315--1326, 2015.

\bibitem{fax2004information}
J.~A. Fax and R.~M. Murray, ``Information flow and cooperative control of
  vehicle formations,'' \emph{IEEE Transactions on Automatic Control}, vol.~49,
  no.~9, pp. 1465--1476, 2004.

\bibitem{pukelsheim1994three}
F.~Pukelsheim, ``The three sigma rule,'' \emph{The American Statistician},
  vol.~48, no.~2, pp. 88--91, 1994.

\bibitem{li2023topology}
Y.~Li, J.~He, C.~Chen, and X.~Guan, ``Topology inference for network systems:
  Causality perspective and non-asymptotic performance,'' \emph{IEEE
  Transactions on Automatic Control}, to appear, 2023.

\bibitem{tikhomirov2016smallest}
K.~E. Tikhomirov, ``The smallest singular value of random rectangular matrices
  with no moment assumptions on entries,'' \emph{Israel Journal of
  Mathematics}, vol. 212, no.~1, pp. 289--314, 2016.

\bibitem{davidson2001local}
K.~R. Davidson and S.~J. Szarek, ``Local operator theory, random matrices and
  banach spaces,'' \emph{Handbook of the geometry of Banach spaces}, vol.~1,
  no. 317-366, p. 131, 2001.

\bibitem{hall2015lie}
B.~Hall, \emph{Lie Groups, Lie Algebras, and Representations: An Elementary
  Introduction}.\hskip 1em plus 0.5em minus 0.4em\relax Springer, 2015, vol.
  222.

\bibitem{chen2000reconstruction}
T.~Chen and D.~Miller, ``Reconstruction of continuous-time systems from their
  discretizations,'' \emph{IEEE Transactions on Automatic Control}, vol.~45,
  no.~10, pp. 1914--1917, 2000.

\bibitem{ding2009reconstruction}
F.~Ding, L.~Qiu, and T.~Chen, ``Reconstruction of continuous-time systems from
  their non-uniformly sampled discrete-time systems,'' \emph{Automatica},
  vol.~45, no.~2, pp. 324--332, 2009.

\bibitem{molloy2022inversea}
T.~L. Molloy, J.~Inga~Charaja, S.~Hohmann, and T.~Perez, \emph{Inverse
  {{Optimal Control}} and {{Inverse Noncooperative Dynamic Game Theory}}: {{A
  Minimum-Principle Approach}}}, ser. Communications and {{Control
  Engineering}}.\hskip 1em plus 0.5em minus 0.4em\relax Springer International
  Publishing, 2022.

\end{thebibliography}
\end{document}